\documentclass[a4paper,10pt]{article}
\usepackage[utf8]{inputenc}

\usepackage{color}

\usepackage{amssymb}
\usepackage{amsmath}
\usepackage{amsfonts}
\usepackage{amsthm}
\usepackage{url}

\usepackage[dvips]{epsfig}
\usepackage[dvips]{graphics}

\usepackage{algorithm}
\usepackage{algorithmic}

\usepackage[all,knot,poly,arc]{xy}
\newcommand{\NN}{{\Bbb N}}

\newcommand{\ZZ}{{\Bbb Z}}

\newcommand{\im}{\mbox{\rm im}}

\newcommand{{\uk}}{\mbox{$\underline{k}$}}

\newcommand{\End}{\mbox{\rm End}}

\newcommand{\id}{\mbox{\rm id}}

\def\nod(#1,#2){\put(#1,#2){\circle*{.125}}
\put(#1,#2){\makebox(0,0.5){{\small$#2$}}}}%
\def\rod(#1,#2){\put(#1,#2){\circle*{.2}}}
\def\NOD(#1,#2)#3{\put(#1,#2){\circle*{.2}}\put(#1,#2){\makebox(0,0.8){{\small$#3$}}}}

\newcounter{exampleNo}

\newtheorem{theorem}{Theorem}[section]
\newtheorem{lemma}[theorem]{Lemma}
\newtheorem{proposition}[theorem]{Proposition}
\newtheorem{corollary}[theorem]{Corollary}

\newenvironment{example}[1][Example \arabic{exampleNo}.]{\begin{trivlist}\refstepcounter{exampleNo}
\item[\hskip \labelsep {\bfseries #1}]}{\end{trivlist}}

\title{On Krohn-Rhodes Theory for Semiautomata}
\author{Karl-Heinz Zimmermann\footnote{Email: k.zimmermann@tuhh.de}\\
Department of Computer Engineering \\
Hamburg University of Technology\\
21071 Hamburg, Germany}

\begin{document}
\maketitle
\begin{abstract}
Krohn-Rhodes theory encompasses the techniques for the study of finite automata 
and their decomposition into elementary automata.
The famous result of Krohn and Rhodes roughly states that 
each finite automaton can be decomposed into elementary components which correspond to
permutation and reset automata connected by a cascade product. 
However, this outcome is not easy to access for the working computer scientist.
%The result of Krohn and Rhodes provides a deep connection between finite automata and semigroups. 
This paper provides a short introduction into Krohn-Rhodes theory based on the valuable work of Ginzburg.
\end{abstract}
\medskip

\mbox{\bf AMS Subject Classification:} 68Q70, 20M35, 20A05
\medskip

\mbox{\bf Keywords:} Semiautomaton, monoid, group, Krohn-Rhodes theory

\section{Introduction}

Automata theory is a subfield of theortical computer science 
and is the study of abstract automata and their computational problems~\cite{eilenberg, mihov, salomaa}.
Algebraic automata theory or Krohn-Rhodes theory
provides a framework for the decomposition of finite automata into elementary components.
More specifically, a finite automaton can be emulated by a feed-forward cascade of automata 
whose transition monoids are simple groups and automata which are banks of reset flip-flops.
This deep result known as Krohn-Rhodes theory was a joint work 
of Kenneth Krohn\index{Krohn, Kenneth} and John Rhodes\index{Rhodes, John}~\cite{kr}
and can be seen as a generalization of the Jordan-H\"older decomposition of finite groups.

This paper provides a short introduction to Krohn-Rhodes theory based on the valuable work of Ginzburg~\cite{ginzburg}.
The inclined reader may note that there is also a purely algebraic representation of Krohn-Rhodes theory by
the recent work of Diekert et al.~\cite{kufi}.
%In this chapter, the Krohn-Rhodes theory is introduced 

\section{Semigroups and Monoids}
In abstract algebra, a semigroup is an algebraic structure with an associative binary operation.
In particular, a monoid is a semigroup with an identity element.
Semigroups and monoids play an important role in the fundamental aspects of theoretical computer science.
Notably, the full transformation monoid captures the notion of function composition, and 
each finitely generated free monoid can be represented by the set of strings over a finite alphabet,
%the syntactic monoid can be used to describe finite state automata, 
%the trace monoid forms a basis for process calculi,
%and the history monoid furnishes a way to depict concurrent computation.
%This appendix provides the background on semigroups and monoids required for reading the text.

A {\em semigroup\/}\index{semigroup} is an algebraic structure consisting of a non-empty set $S$ and a binary operation 
\begin{eqnarray}
\cdot:S\times S\rightarrow S:(x,y)\mapsto x\cdot y
\end{eqnarray}
that satisfies the {\em associative law\/}\index{associative law}; i.e., 
for all $r,s,t\in S$, 
\begin{eqnarray}
(r\cdot s) \cdot t = r\cdot (s\cdot t).
\end{eqnarray}
A semigroup $S$ is {\em commutative\/}\index{semigroup!commutative} if it satisfies the {\em commutative law\/}\index{commutative law}; i.e., 
for all $s,t\in S$, 
\begin{eqnarray}
s\cdot t = t\cdot s.
\end{eqnarray}
A semigroup $S$ with operation $\cdot$ is also written as a pair $(S,\cdot)$.
The juxtaposition~$s\cdot t$ is also denoted as $st$.
The number of elements of a semigroup $S$ is called the {\em order}\index{order} of $S$.

\begin{example}
\mbox{ }
\begin{itemize} 
\item
Each singleton set $S=\{s\}$ forms a commutative semigroup under the operation $s\cdot s = s$.
\item
The set of natural numbers $\NN$ forms a commutative semigroup under addition, or multiplication.
\item
The set of integers $\ZZ$ forms a commutative semigroup under minimum, or maximum.
%\item
%The set of all non-empty finite strings $\Sigma^+$ over an alphabet $\Sigma$ forms a semigroup under concatenation,  
%\begin{eqnarray}
%\circ:\Sigma^+\times \Sigma^+\rightarrow \Sigma^+:(u,v)\mapsto u\circ v =uv.
%\end{eqnarray}
%called {\em free semigroup\/}\index{semigroup!free} over $\Sigma$.
%The semigroup $\Sigma^+$ is commutative only if the alphabet $\Sigma$ is a singleton set.
\end{itemize}
\end{example}

A {\em monoid\/}\index{monoid} is a semigroup $M$ with a distinguished element $e\in M$, called {\em identity element\/}\index{identity element}, which satisfies for all $m\in M$,
\begin{eqnarray}
m\cdot e = m = e\cdot m .
\end{eqnarray}
A monoid is {\em commutative\/}\index{monoid!commutative} if its operation is commutative.
A monoid $M$ with operation $\cdot$ and identity element $e$ is also written as $(M,\cdot,e)$.

Note that each monoid $M$ has a unique identity element,
since if $e$ and $e'$ are identity elements of $M$, then $e=e\cdot e' = e'$.
\begin{example}
\mbox{ }
\begin{itemize}
\item
Any semigroup $S$ can be made into a monoid by adjoining an element $e$ not in $S$ and defining $s\cdot e = s = e\cdot s$ for all $s\in S\cup\{e\}$.
%\item
%The semigroup $\Sigma^+$ can be turned into the monoid $\Sigma^*=\Sigma^+\cup\{\epsilon\}$ by adjoining the empty word $\epsilon$.
%The monoid $\Sigma^*$ is commutative only if the alphabet $\Sigma$ is the empty set, in which case $\Sigma^* = \{\epsilon\}$, or a singleton set.
\item 
The set of natural numbers $\NN_0$ forms a commutative monoid under addition (with identity element~0), or multiplication (with identity element~1).
\item
The power set $P(X)$ of a set $X$ forms a commutative monoid under intersection (with identity element $X$), or union (with identity element $\emptyset$).
%In particular, for each subset $A$ of $X$.
%\begin{eqnarray}
%A\cap X = A\quad \mbox{und}\quad A\cup \emptyset = A.
%\end{eqnarray}
\item
The set of all integral $2\times 2$ matrices $\ZZ^{2\times 2}$ forms a commutative monoid under addition (with the zero matrix $O$ as identity element),
and a non-commutative monoid under multiplication (with the unit matrix $I$ as identity element).
%The set of all integral $2\times 2$ matrices 
%\begin{eqnarray}
%\ZZ^{2\times 2} = \left\{ \left( \begin{array}{cc}  a&b\\c&d \end{array} \right)\mid a,b,c,d\in\ZZ \right\}
%\end{eqnarray}
%forms a commutative monoid under matrix addition 
%\begin{eqnarray}
%\left( \begin{array}{cc}  a&b\\c&d \end{array} \right) + \left( \begin{array}{cc}  a'&b'\\c'&d' \end{array} \right) 
%= \left( \begin{array}{cc}  a+a'&b+b'\\c+c'&d+d' \end{array} \right)
%\end{eqnarray}
%with the zero matrix $O$ as identity element, or under matrix multiplication 
%\begin{eqnarray}
%\left( \begin{array}{cc}  a&b\\c&d \end{array} \right) \cdot \left( \begin{array}{cc}  e&f\\g&h \end{array} \right) 
%= \left( \begin{array}{cc}  ae+bg&af+bh\\ce+dg&cf+dh \end{array} \right)
%\end{eqnarray}
%with the unit matrix $I$ as identity element.
%The first monoid is commutative and the second one is not.
\end{itemize} 
\end{example}
In view of the above observation, we will concentrate on monoids instead of semigroups in the sequel.

\begin{proposition}\label{p-MX}
Let $(M,\cdot,e)$ be a monoid and $X$ a non-empty set.
The set of all mappings from $M$ to $X$,
\begin{eqnarray}
M^X = \{f\mid f:X\rightarrow M\}, 
\end{eqnarray}
forms a monoid under component-wise multiplication
\begin{eqnarray}
(f\cdot g) (x) = f(x)g(x),\quad f,g\in M^X, x\in X.
\end{eqnarray}
\end{proposition}
\begin{proof}
The operation is associative, 
since for all $f,g,h\in M^X$ and $x\in X$,
\begin{eqnarray*}
((f\cdot g)\cdot  h)(x) &=& (f\cdot g)(x)h(x) = (f(x)g(x)) h(x) = f(x) (g(x)h(x)) \\
&=& f(x)(g\cdot  h)(x) = (f\cdot (g\cdot  h))(x),
\end{eqnarray*}
and the mapping $X\rightarrow M:x\mapsto e$ is the identity element.
\end{proof}
The image of $x$ under a mapping $f$ is denoted by $f(x)$ or $xf$.

\begin{proposition}
Let $X$ be a non-empty set. 
The set of all mappings on $X$,
\begin{eqnarray}
T(X) = \{f\mid f:X\rightarrow X\},
\end{eqnarray}
forms a monoid under function composition 
\begin{eqnarray}
(fg)(x) = g(f(x)),\quad x\in X,
\end{eqnarray}
and the identity function $\id=\id_X:X\rightarrow X:x\mapsto x$ is the identity element.
\end{proposition}
\begin{proof}
The function composition is associative and the identity function is the identity element.
\end{proof}
The monoid $T(X)$ is called the {\em full transformation monoid\/}\index{full transformation monoid} of $X$.

\begin{proposition}
Let $(M,\cdot,e)$ and $(M',\cdot',e')$ be monoids.
The direct product set $M\times M'$ forms a monoid under the component-wise operation
\begin{eqnarray}
(x,x')\circ (y,y') = (x\cdot y,x'\cdot' y'),\quad x,y\in M,\, x',y'\in M',
\end{eqnarray}
and the pair $(e,e')$ is the identity element.
The monoid $M\times M'$ is commutative if $M$ and $M'$ are commutative.
\end{proposition}
\begin{proof}
Let $x,y,z\in M$ and $x',y',z'\in M'$.
The operation is associative, since 
\begin{eqnarray*}
(x,x')\circ [(y,y')\circ(z,z')] 
&=& (x,x')\circ (y\cdot z,y'\cdot'z') \\
&=& (x\cdot(y\cdot z),x'\cdot'(y'\cdot'z')) \\
&=& ((x\cdot y)\cdot z,(x'\cdot'y')\cdot'z') \\
&=& (x\cdot y,x'\cdot' y') \circ (z,z')\\
&=& [(x,x')\circ (y,y')]\circ(z,z'). 
\end{eqnarray*}
The operation is commutative, since if $M$ and $M'$ are commutative, then
\begin{eqnarray*}
(x,x')\circ (y,y') = (x\cdot y,x'\cdot' y') = (y\cdot x,y'\cdot' x') = (y,y')\circ (x,x').
\end{eqnarray*}
Finally, the identity element is $(e,e')$, since
\begin{eqnarray*}
(x,x')\circ (e,e') &=& (x\cdot e,x'\cdot' e') = (x,x'), \\
(e,e')\circ (x,x') &=& (e\cdot x,e'\cdot' x') = (x,x'). 
\end{eqnarray*}
\end{proof}

An {\em alphabet\/}\index{alphabet} denotes always a finite non-empty set, 
and the elements of an alphabet are called {\em symbols\/}\index{symbol} or {\em characters\/}\index{character}.
A {\em word\/}\index{word!see string} or {\em string\/}\index{string} {\em of length\/}\index{string!length} $n\geq 0$ over an alphabet $\Sigma$ is a finite sequence 
$x=(x_1,\ldots,x_n)$ with components in $\Sigma$.
A string $x=(x_1,\ldots,x_n)$ is usually written without delimiters $x =x_1\ldots x_n$.
If~$n=0$, then $x$ is the {\em empty string\/}\index{string!empty} denoted by $\epsilon$.
\begin{proposition}
Let $\Sigma$ be an alphabet.
The set $\Sigma^*$ of all finite strings over $\Sigma$ forms a monoid under concatenation of strings
%\begin{eqnarray}
%(x_1\ldots x_m)\circ(y_1,\ldots,y_n) = x_1\ldots x_my_1,\ldots,y_n,
%\end{eqnarray}
and the empty string $\epsilon$ is the identity element.
\end{proposition}
%F"ur die W"orter $x={\tt face}$ und $y={\tt book}$ "uber dem lateinischen Alphabet $\Sigma$ gilt 
%$$x\circ y = {\tt facebook} \quad\mbox{und}\quad y\circ x = {\tt bookface}.$$ 
%\begin{proof}
%The concatenation of strings is associative, since for all $x,y,z\in\Sigma^*$, 
%\begin{eqnarray*}
%x\circ(y\circ z) &=& x_1\ldots x_l(y_1\ldots y_m z_1\ldots z_n)\\
%%&=& x_1\ldots x_ly_1\ldots y_m z_1\ldots z_n\\
%&=& (x_1\ldots x_ly_1\ldots y_m) z_1\ldots z_n\\
%&=& (x\circ y)\circ z.
%%%\end{eqnarray*}
%Moreover, the empty string satisfies for all $x\in\Sigma^*$, $x\circ \epsilon = x = \epsilon\circ x$.
%\end{proof}
The monoid $\Sigma^*$ is called the {\em word monoid\/}\index{word monoid} over $\Sigma$.

%\begin{example}
%Let $\Sigma=\{1\}$.
%The unary word monoid $\Sigma^*$ is commutative and consists of all strings 
%$$\epsilon, \;1, \;11, \;111, \;1111,\;11111,\ldots .$$
%Let $\Sigma=\{0,1\}$.
%The binary word monoid $\Sigma^*$ is non-commutative and contains the strings
%$$\epsilon,0,1,00,01,10,11,000,001,010,011,100,101,110,111,0001,\ldots .$$
%%$$\begin{array}{c||l}
%%\mbox{length} & \mbox{strings}\\\hline\hline
%%%0 & \epsilon\\
%%1 & 0,\,1,\\
%%2 & 00,\, 01, \, 10, \, 11,\\
%%3 & 000,\, 001, \, 010, \, 011,\, 100,\, 101, \, 110, \, 111,\\
%%4 & 0000,\ldots
%%\end{array}$$
%\end{example}

A {\em submonoid\/}\index{submonoid} of a monoid $(M,\cdot,e)$ is a non-empty subset $U$ of $M$ which contains the identity element $e$ 
and is closed under the monoid operation; i.e.,
for all $u,v\in U$, $u\cdot v\in U$.

\begin{example}
\mbox{ }
\begin{itemize}
\item
Each monoid $(M,\cdot,e)$ has two trivial submonoids, $M$ and $\{e\}$.
%\item
%Let $m$ be a nonnegative natural number.
%The set of all multiples of $m$, $m\NN = \{k\cdot m\mid k\in\NN\}$ forms a submonoid of the monoid $\NN$ under addition, or multiplication.
\item
Let $\Sigma_1$ and $\Sigma_2$ be alphabets.
If $\Sigma_1$ is a subset of $\Sigma_2$, then $\Sigma_1^*$ is a submonoid of $\Sigma_2^*$.
\item
Consider the set of all upper triangular integral $2\times 2$ matrices 
$U = \left\{ \left( \begin{array}{cc}  a&b\\0&c \end{array} \right)\mid a,b,c\in\ZZ \right\}$.
The set $U$ forms a submonoid of the matrix monoid $(\ZZ^{2\times 2},+, O)$, 
since $U$ is closed under matrix addition 
$\left(\begin{array}{cc}a&b\\0&c \end{array}\right) + \left(\begin{array}{cc}  d&e\\0&f \end{array}\right)
= \left(\begin{array}{cc}a+d&b+e\\0&c+f \end{array}\right)   $
and contains the unit element $O$.
Likewise, the set $U$ forms a submonoid of the matrix monoid $(\ZZ^{2\times 2},\cdot, I)$, 
since $U$ is closed under matrix multiplication 
$\left(\begin{array}{cc}a&b\\0&c \end{array}\right) \left(\begin{array}{cc}  d&e\\0&f \end{array}\right)
= \left(\begin{array}{cc}ad&ae+bf\\0&cf \end{array}\right)   $
and contains the unit element $I$.
\end{itemize} 
\end{example}

\begin{lemma}\label{p-right}
Let $(M,\cdot,e)$ be a monoid.
The mapping 
%$\ell_x:M\rightarrow M:m\mapsto xm$
$r_x:M\rightarrow M:m\mapsto mx$
is called {\em right multiplication\/} with $x\in M$.
The set of all right multiplications with elements from $M$,
$R_M = \{r_x\mid x\in M\}$,
forms a submonoid of the full transformation monoid $T(M)$.
\end{lemma}
\begin{proof}
Let $x,y,z\in M$.
Then $r_xr_y = r_{xy}$ since for all $z\in M$, 
$$zr_{xy} =  z(xy)=(zx)y = (zx)r_y = (zr_x)r_y = z(r_x r_y),$$
and $r_e=\id$ since for all $z\in M$, $r_e(z) = ze = z =\id(z).$
\end{proof}
The left multiplications are similarly defined.
Each submonoid of a full transformation monoid is called a {\em transformation monoid}\index{transformation monoid}.

%\begin{example}
%Consider the word monoid $\Sigma^*$ over the latin alphabet $\Sigma$.
%The word ${\tt face}$ induces the left multiplication $r_{\tt face}:\Sigma^*\rightarrow\Sigma^*$ 
%with
%$r_{{\tt face}} ({\tt book}) = {\tt facebook}$
%and
%$r_{{\tt face}} ({\tt it}) = {\tt faceit}$.
%\end{example}

\begin{proposition}
Let $M$ be a monoid and $X$ be a subset of $M$.
The set of all finite products of elements of~$X$,
\begin{eqnarray}
\langle X\rangle = \{x_1\cdots x_n\mid x_1,\ldots,x_n\in X, n\geq 0\},
\end{eqnarray}
forms a submonoid of $M$.
It is the smallest submonoid of $M$ containing $X$.
\end{proposition}
\begin{proof}
The product of two finite products of elements of $X$ is also a finite product, and
the empty product $(n=0)$ gives the identity element of $M$.
Thus~$\langle X\rangle$ is a submonoid of $M$.
Moreover, each submonoid $U$ of $M$ containing the set~$X$ 
must also contain all finite products of elements of $X$; i.e., we have $\langle X\rangle \subseteq U$.
\end{proof}
For each finite set $X=\{x_1,\ldots,x_n\}$ write $\langle x_1,\ldots,x_n\rangle$ instead of $\langle X\rangle$.

Let $M$ be a monoid and $X$ be a subset of $M$.
Then $M$ is said to be {\em generated by\/} $X$ if $M=\langle X\rangle$. 
In this case, $X$ is a {\em generating set\/}\index{generating set} of $M$.
In particular, $M$ is {\em finitely generated\/}\index{monoid!finitely generated} if $M$ has a finite generating set $X$.

\begin{example}
\mbox{ }
\begin{itemize}
\item
For each monoid $(M,\cdot,e)$, we have $\langle \emptyset\rangle = \{e\}$ und $\langle M\rangle = M$.
\item
The monoid $(\NN_0,+,0)$ is generated by $X=\{1\}$, since $n = n\cdot 1= 1+\ldots+1$ for all $n\in\NN_0$.
\item
The monoid $(\NN,\cdot,1)$ is generated by the set of prime numbers, since by the fundamental theorem of arithmetic, each positive natural number 
is a (unique) product of prime numbers. % this product is uniquely determined.
\item
The word monoid $\Sigma^*$ is generated by the underlying alphabet $\Sigma$.
\item
%Erzeugenden <X>
Consider the matrix monoid $(\ZZ^{2\times 2},\cdot, I)$.
The powers of the matrix 
$A = \left(\begin{array}{cc}0&1\\-1&0\end{array}\right)$
are 
$A^2 = \left(\begin{array}{cc}-1&0\\0&-1\end{array}\right)$,
$A^3 = \left(\begin{array}{cc}0&-1\\1&0\end{array}\right)$,
$A^4 = \left(\begin{array}{cc}1&0\\0&1\end{array}\right)$,
and $A^5=A$.
Thus $\langle A\rangle = \{I,A,A^2,A^3\}$ is a (cyclic) submonoid of order~4.
\end{itemize}
\end{example}

%\begin{example}
%\mbox{ }
%\begin{itemize}
%\item
%For each natural number $m\geq 1$, the submonoid $\langle m\rangle = \{k\cdot m\mid k\in\NN_0\}$ of $(\NN_0,+,0)$ consists of all multiples of $m$.
%\item
%%For each natural number $m\geq 1$, the submonoid $\langle m\rangle = \{m^k\mid k\in\NN_0\}$ of $(\NN,\cdot,1)$ consists of all powers of $m$.
%\item
%For each alphabet $\Sigma$, the submonoid $\langle \Sigma\rangle$ of $\Sigma^*$ is equal to the monoid $\Sigma^*$.
%\end{itemize}
%\end{example}

Let $(M,\cdot,e)$ and $(M',\cdot',e')$ be monoids.
A {\em homomorphism\/}\index{homomorphism} is a mapping $\phi:M\rightarrow M'$ which assigns the identity elements 
\begin{eqnarray}
\phi(e) = e'
\end{eqnarray}
and commutes with the monoid operations, i.e., for all $x,y\in M$,
\begin{eqnarray}
\phi(x\cdot y) = \phi(x)\cdot' \phi(y).
\end{eqnarray}

A homomorphism $\phi:M\rightarrow M'$ between monoids is a {\em monomorphism}\index{monomorphism} if $\phi$ is one-to-one,
an {\em epimorphism}\index{epimorphism} if $\phi$ is onto,
an {\em isomorphism}\index{isomorphism} if $\phi$ is one-to-one and onto,
an {\em endomorphism}\index{endomorphism} if $M=M'$, and
an {\em automorphism}\index{automorphism} if $\phi$ is an endomorphism and an isomorphism.
In particular, two monoids $M$ and $M'$ are {\em isomorphic}\index{isomorphic} if there is an isomorphism $\phi:M\rightarrow M'$.
Isomorphic monoids have the same structure up to a relabelling of the elements.

\clearpage 
\begin{example}\label{e-hom-mon}
\mbox{ }
\begin{itemize}
%\item
%The power mapping
%\begin{eqnarray}
%\phi:\NN_0\rightarrow\NN:n\mapsto 2^n
%\end{eqnarray}
%is a homomorphism from $(\NN_0,+,0)$ into $(\NN,\cdot,1)$, 
%since for all $m,n\in\NN_0$,
%$$\phi(m+n) = 2^{m+n} = 2^m\cdot 2^n = \phi(m)\cdot \phi(n)$$
%and $\phi(0) = 2^0 = 1.$ 
%This homomorphism is a monomorphism.
\item
Let $\Sigma$ be an alphabet.
The length mapping 
$\phi:\Sigma^*\rightarrow\NN_0:x\mapsto |x|$
is a homomorphism from $(\Sigma^*,\cdot,\epsilon)$ onto $(\NN_0,+,0)$,
since for all $x,y\in \Sigma^*$,
$\phi(xy) =|xy| = |x|+|y| = \phi(x)+\phi(y)$
and $\phi(\epsilon) = |\epsilon| = 0 $.
This homomorphism is an epimorphism.
\item
The mapping
$\phi: \ZZ^{2\times 2} \rightarrow \ZZ: A \mapsto  \det(A)$,
where $\det(A)=ad-bc$ denotes the determinant of the matrix $A= \left( \begin{array}{cc}  a&b\\c&d \end{array} \right)$, 
is a homorphisms from $(\ZZ^{2\times 2},\cdot,I)$ onto $(\ZZ,\cdot,1)$.
This homomorphism is an epimorphism.
\item
Let $X$ be a set.
The mapping
$\phi: P(X)\rightarrow P(X): A \mapsto  X\setminus A$
is a homomorphism from $(P(X),\cap,X)$ onto $(P(X),\cup,\emptyset)$,
since for all $A,B\in P(X)$,
$\phi(A\cap B) = X\setminus (A\cap B) = (X\setminus A)\cup (X\setminus B) =\phi(A)\cup \phi(B)$
and $\phi(X) = X\setminus X = \emptyset$.
This homomorphism is an isomorphism.
%\item
%The mapping
%\begin{eqnarray}
%\phi: \NN_0 \rightarrow \ZZ^{2\times2}: n\mapsto \left( \begin{array}{cc}  1&n\\0&1 \end{array} \right)
%\end{eqnarray}
%is a homomorphism from $(\NN_0,+,0)$ into $(\ZZ^{2\times 2},\cdot, I)$,
%since for all $m,n\in\NN_0$,
%$$\phi(m+n) 
%= \left( \begin{array}{cc}  1&m+n\\0&1 \end{array} \right) 
%= \left( \begin{array}{cc}  1&m\\0&1 \end{array} \right) \left( \begin{array}{cc}  1&n\\0&1 \end{array} \right) 
%=\phi(m)+\phi(n)$$
%and
%$\phi(0) = I$.
%This homomorphism is a monomorphism.
%The set of all upper triangular matrices \begin{eqnarray} U_0 = \left\{ \left( \begin{array}{cc}  1&n\\0&1 \end{array} \right)\mid n\in\NN_0 \right\} \end{eqnarray} forms a submonoid of $(\ZZ^{2\times 2},\cdot, I)$.  Since it is closed under matrix multiplication 
\end{itemize}
\end{example}

\begin{lemma}\label{l-comp}
Let $\phi:M\rightarrow M'$ be a homomorphism between monoids.
\begin{itemize}
\item
The image of $\phi$ is a submonoid of $M'$, written  $\im(\phi)$.
%\begin{eqnarray} \im(\phi) = \{\phi(x)\mid x\in M\} \end{eqnarray}
\item
If $N'$ is a submonoid of $M'$, then $\phi^{-1}(N')$ is a submonoid of $M$.
\item
If $\psi:M'\rightarrow M''$ is another homomorphism between monoids, the composition $\phi\psi:M\rightarrow M''$ is also a homomorphism.
\end{itemize}
\end{lemma}
\begin{proof}
Let $x,y\in M$.
Then $\phi(x)\cdot'\phi(y) = \phi(x\cdot y) \in\im(\phi)$.
Moreover, $e'= \phi(e)\in \im(\phi)$.
Therefore, $\im(\phi)$ is a submonoid of $M'$.

Let $x,y\in\phi^{-1}(N')$.
Then $\phi(x\cdot y) = \phi(x)\cdot' \phi(y) \in N'$. Moreover,  $\phi(e)=e'\in N'$ and so $e\in\phi^{-1}(N')$.
Thus $\phi^{-1}(N')$ is a submonoid of $M$. 

The composition of mappings is a mapping.
Moreover, for all $x,y\in M$,
\begin{eqnarray*}
(\phi\psi)(x\cdot y) 
&=& \psi(\phi(x\cdot y)) = \psi(\phi(x)\cdot'\phi(y)) \\
&=& \psi(\phi(x))\cdot'' \psi(\phi(y))\\
&=& (\phi\psi)(x)\cdot'' (\phi\psi)(y)
\end{eqnarray*}
and $(\phi\psi)(e) = \psi(\phi(e)) = \psi(e') = e''.$
\end{proof}
%The monoid $\im(\phi)$ is called the {\em image\/}\index{image} of $\phi$, and
%the monoid $\ker(\phi)$ is called the {\em kernel\/}\index{kernel} of $\phi$.
%Note that $\ker(\phi)$ can be the trivial monoid $\{e\}$ even if the homomorphism $\phi$ is not a monomorphism 
%(e.g., length function in Example~\ref{e-hom-mon}).

\begin{proposition}
The set of all endomorphisms $\End(M)$ of a monoid $M$ forms a monoid under the composition of mappings.
\end{proposition}
\begin{proof}
By Lemma~\ref{l-comp}, the composition of two endomorphisms 
$\phi:M\rightarrow M$ and $\psi:M\rightarrow M$ is an endomorphism $\phi\psi:M\rightarrow M$.
Thus the set of all endomorphisms of $M$ is closed under composition.
Moreover, the composition of mappings is associative and the identity mapping $\id:M\rightarrow M$ is the identity element.
\end{proof}
Therefore, the endomorphism monoid $\End(M)$ is a submonoid of the full transformation monoid $T(M)$.

\begin{proposition}
Let $M$ and $N$ be monoids.
If $\phi:M\rightarrow M'$ is an isomorphism, the inverse mapping $\phi^{-1}:M'\rightarrow M$ is also an isomorphism.
\end{proposition}
\begin{proof}
The inverse mapping $\phi^{-1}$ exists, since $\phi$ is bijective.
Moreover, the mapping $\phi^{-1}$ is also bijective and we have $\phi\phi^{-1}=\id = \phi^{-1}\phi$.
Since $\phi$ is a homomorphism, we have for all $x,y\in M'$,
$$\phi(\phi^{-1}(x) \cdot \phi^{-1}(y)) = \phi(\phi^{-1}(x)) \cdot' \phi(\phi^{-1}(y)) = x\cdot' y$$ 
and therefore
$\phi^{-1}(x) \cdot \phi^{-1}(y) = \phi^{-1}(x\cdot' y)$.
Moreover, $\phi(e)=e'$ implies $\phi^{-1}(e')=e$.
\end{proof}

\begin{theorem}\label{t-right}
Each monoid $M$ is isomorphic to a transformation monoid. % submonoid of the full transformation monoids $T_M$.
\end{theorem}
\begin{proof}
Claim that the mapping $\phi: M\rightarrow R_M : x\mapsto r_x$ is an isomorphism.
Indeed, for all $x,y\in M$, we have $\phi(xy) = r_{xy} = r_xr_y = \phi(x)\phi(y)$
and $\phi(e) = r_e = \id.$
It is clear that the mapping is onto.
It is also one-to-one, since $r_x=r_y$ with $x,y\in M$ implies by substitution of the identity element $e$,
$x  = xe = r_x(e) = r_y(e) =ye = y.$
\end{proof}
Therefore, the theory of monoids can be considered as the theory of transformations.

\begin{example}
Consider the monoid $M=\{a,e,b\}$ (corresponding to the multiplicative monoid $\{-1,0,1\}$) with the multiplication table
$$\begin{array}{r||rrr}
\cdot & a & e & b\\\hline\hline
a     & b & e & a\\
e     & e & e & e\\
b     & a & e & b
\end{array}$$
The transformation monoid of left multiplications $R_M$ consists of the elements 
$$r_b = \left( \begin{array}{rrr} a & e & b\\ b & e & a\end{array} \right),\;
r_e   = \left( \begin{array}{rrr} a & e & b\\ e & e & e\end{array} \right),\;
r_b   = \left( \begin{array}{rrr} a & e & b\\ a & e & b\end{array} \right),$$
where the right multiplications correspond one-to-one with the rows of the multiplication table of $M$.
The multiplication table of $R_M$ is as follows:
$$\begin{array}{r||rrr}
\cdot  & r_b & r_e & r_b\\\hline\hline
r_a    & r_b & r_e & r_a\\
r_e    & r_e & r_e & r_e\\
r_b    & r_a & r_e & r_b
\end{array}$$
\end{example}

Let $M$ be a monoid.
Suppose $\equiv$ is an equivalence relation on $M$.
The {\em equivalence class\/}\index{equivalence class} of $x\in M$ is given by the set of all elements of $M$ which are equivalent to $x$,
\begin{eqnarray}
[x] = \{y\in M\mid x\equiv y\}.
\end{eqnarray}
The {\em quotient set\/}\index{quotient set} consisting of all equivalence classes is
\begin{eqnarray}
M/\!\equiv \;= \{[x]\mid x\in M\}.
\end{eqnarray}
An equivalence relation $\equiv$ on $M$ is a {\em congruence relation}\index{congruence} 
if it is compatible with the monoid operation, i.e., for all $x,x',y,y'\in M$,
\begin{eqnarray}
x\equiv x' \wedge y\equiv y' \;\Longrightarrow\; xy\equiv x'y',
\end{eqnarray}
or equivalently,
\begin{eqnarray}
[x] = [x'] \wedge [y] = [y'] \;\Longrightarrow\; [xy] = [x'y'].
\end{eqnarray}
The equivalence classes of a congruence relation are called {\em congruence classes}\index{congruence class} and elements in the same congruence class are said to be
{\em congruent}\index{congruent}.

\begin{example}
Let $n\geq 2$ be an integer.
The relation on $\ZZ$ given by
$$a\equiv b \;(\mod\;n)\quad :\Longleftrightarrow\quad n\mbox{ divides } (a-b)$$
is a congruence on the monoid $(\ZZ,+,0)$, or the monoid $(\ZZ,\cdot,1)$.
\end{example}

\begin{proposition}\label{p-Mc}
Let $M$ be a monoid and $\equiv$ a congruence relation on $M$.
Then the quotient set $M/\!\equiv$ forms a monoid under the operation
\begin{eqnarray}
[x] \cdot [y] = [xy], \quad x,y\in M.
\end{eqnarray}
The monoid $M/\!\equiv$ is commutative if $M$ is commutative.
\end{proposition}
\begin{proof}
The operation is well-defined, since if $[x]=[x']$ and $[y]=[y']$ for some $x,x',y,y'\in M$, 
then by the congruence property $[xy] = [x'y']$.
The operation is associative, since  for all $x,y,z\in M$, 
$$[x]\cdot([y]\cdot[z]) = [x]\cdot[yz] = [x(yz)] = [(xy)z] = [xy]\cdot[z] = ([x]\cdot[y])\cdot[z].$$
The equivalence class of the identity element $e\in M$ is the identity element of~$M/\!\equiv$, 
since for all $x\in M$, $[x]\cdot[e] = [xe] = [x]=[ex] = [e]\cdot[x]$.
Finally, if $M$ is commutative, then for all $x,y\in M$, $[x]\cdot[y] = [xy] = [yx] = [y]\cdot[x]$.
\end{proof}

\begin{example}
Reconsider the congruence modulo $n$ on $\ZZ$.
Two integers are congruent modulo $n$ if and only if they have the same remainder upon integral division by $n$.
The quotient set $\ZZ/\!\equiv$, consisting of all congruence classes of remainders $[0],\ldots,[n-1]$, 
forms a commutative monoid under addition, or multiplication. 
\end{example}

\begin{proposition}
Let $M$ be a monoid and $\equiv$ a congruence on $M$.
Then the mapping $\pi:M\rightarrow M/\!\equiv:x\mapsto [x]$ is an epimorphism.
\end{proposition}
\begin{proof}
The mapping $\pi$ is surjective.
Moreover, by Prop.~\ref{p-Mc}, for all $x,y\in M$, $\phi(xy) = [xy] = [x]\cdot[y] = \phi(x)\cdot\phi(y)$,
and $\phi(e) = [e]$.
\end{proof}

\begin{example}
Consider the monoid $M$ given by the multiplication table 
%Klein Vier
$$\begin{array}{c||cccc} 
\cdot & e & a & b & c\\ \hline\hline
e     & e & a & b & c \\
a     & a & e & c & b \\
b     & b & c & e & a \\
c     & c & b & a & e \\
\end{array}$$
The partition $\{\{e,a\},\{b,c\}\}$ defines a congruence $\equiv$ on $M$, and the quotient monoid $M/\!\equiv$ is given by the multiplication table
$$\begin{array}{c||cc} 
\cdot & \mbox{$[e]$} & \mbox{$[b]$} \\ \hline\hline
\mbox{$[e]$} & \mbox{$[e]$} & \mbox{$[b]$}\\
\mbox{$[b]$} & \mbox{$[b]$} & \mbox{$[e]$}\\
\end{array}$$
\end{example}

\begin{lemma}
Let $\phi:M\rightarrow N$ be a homomorphism between monoids $M$ and $N$.
Then the relation on $M$ defined by
\begin{eqnarray}
x \equiv y\quad:\Longleftrightarrow\quad \phi(x)=\phi(y),  \quad x,y\in M,
\end{eqnarray}
is a congruence on $M$.
\end{lemma}
\begin{proof}
Let $x,x',y,y'\in M$ with $x\equiv x'$ and $y\equiv y'$.
Then $\phi(x)=\phi(x')$ and $\phi(y)=\phi(y')$.
Since $\phi$ is a homomorphism, $\phi(xy) = \phi(x)\cdot\phi(y) = \phi(x')\cdot\phi(y') = \phi(x'y')$
and hence $xy\equiv x'y'$.
\end{proof}
This congruence relation is called the {\em kernel}\index{kernel} of $\phi$ and is denoted by $\ker(\phi)$.

\begin{proposition}[Homomorphism Theorem]
Let $\phi:M\rightarrow N$ be a homomorphism between monoids $M$ and $N$.
Then the mapping $\psi: M/\ker(\phi)\rightarrow N$ given by $[x]\mapsto \phi(x)$ is a monomorphism such that $\pi\psi = \phi$,
where $\pi:M\rightarrow M/\ker(\phi):x\mapsto[x]$ is a natural epimorphism.
\end{proposition}
\begin{proof}
The mapping $\psi$ is well-defined, since if $[x]=[y]$ for some $x,y\in M$,
then $x\equiv y$ and so $\phi(x)=\phi(y)$.
The mapping is a homomorphism, since for all $x,y\in M$,
$\psi([xy]) =\phi(xy) = \phi(x)\cdot \phi(y) = \psi([x])\cdot \psi([y])$, and $\psi([e]) = \phi(e) = e'$ for the identity elements $e\in M$ and $e'\in N$.
The mapping is one-to-one, since if $\psi([x]) = \psi([y])$ for some $x,y\in M$,
then $\phi(x)=\phi(y)$,  or equivalently $x\equiv y$ or $[x]=[y]$.
Finally, for each $x\in M$, $(\pi\psi)(x) = \psi(\pi(x)) = \psi([x]) = \phi(x)$.
\end{proof}
This result provides the following commutative diagram:
$$\xymatrix{
M\ar@{->}[d]^{\pi}\ar@{->}[rr]^{\phi} & & N   \\
M/\ker(\phi) \ar@{-->}[rru]_\psi         & & \\
}$$
The Homomorphism theorem shows that any homomorphic image of a monoid can be constructed by means of a congruence relation on $M$.

\begin{example}
Reconsider the epimorphism (length mapping) $\phi:\Sigma^*\rightarrow\NN_0:x\mapsto |x|$.
In the kernel of this mapping, two strings are congruent if and only if they have the same length; 
i.e., the congruence class of $x\in\Sigma^*$ is  $[x] = \{y\in\Sigma^*\mid |x|=|y|\}$.
The monomorphism $\psi:\Sigma^*/\ker(\phi) \rightarrow \NN_0:[x]\mapsto |x|$ assigns to each congruence class the length of its elements.
\end{example}

A monoid $M$ is called {\em free\/}\index{monoid!free} 
if $M$ has a generating set $X$ such that each element of $M$ can be uniquely written as a product of elements in $X$.
In this case, the set $X$ is called a {\em basis\/}\index{monoid!basis} of $M$.
%It is clear that if $X$ is a basis of $M$, then $M=\langle X\rangle$.
%A free module $M$ is {\em finitely generated\/}\index{monoid!finitely generated} if $M$ has a finite basis $X$.
%%Im Monoid $(\NN,\cdot,1)$ gilt etwa die Beziehung $2\cdot 3 = 6 = 3\cdot 2$.
%Freie Monoide spielen in der Theoretischen Informatik eine wichtige Rolle.

\begin{proposition}\label{p-sigma}
The word monoid\/ $\Sigma^*$ over $\Sigma$ is free with basis $\Sigma$.
\end{proposition}
\begin{proof}
By definition, $\Sigma^* = \langle \Sigma\rangle$.
Suppose the word $x\in\Sigma^*$ has two representations
$$x_1x_2\ldots x_m = x = y_1y_2\ldots y_n,\quad x_i,y_j\in\Sigma.$$
Each word is a sequence or mapping  
$$\left( \begin{array}{cccc} 1 & 2 & \ldots & m \\ x_1 & x_2 & \ldots & x_m \end{array} \right) 
= x = 
\left( \begin{array}{cccc} 1 & 2 & \ldots & n \\ y_1 & y_2 & \ldots & y_n \end{array} \right).$$
Thus  by the equality of mappings %$$\{(1,x_1),(2,x_2),\ldots,(m,x_m)\} = \{(1,y_1),(2,y_2),\ldots,(n,y_n)\}.$$ 
$m=n$ and $x_i=y_i$ for $1\leq i\leq n$.
\end{proof}

\begin{theorem}[Universal Property]\label{t-up}
Let $M$ be a free monoid with basis $X$.
For each monoid $M'$ and each mapping $\phi_0:X\rightarrow M'$ 
there is a unique homomorphism $\phi:M\rightarrow M'$ extending $\phi_0$.
%$\alpha (x) = \alpha_0(x)$ f"ur alle $x\in X$.
\end{theorem}
%{\em Fortsetzung:}
%$\alpha$ setzt $\alpha_0$ fort, kurz $\alpha_{|X} = \alpha_0$, bedeutet, 
%dass $\alpha$ und $\alpha_0$ auf $X\subseteq M$ "ubereinstimmen: $\alpha(x)=\alpha_0(x)$ f"ur alle $x\in X$.
The universal property provides the following commutative diagram:
$$\xymatrix{
X\ar@{->}[d]^{\rm id}\ar@{->}[rr]^{\phi_0} & & M'   \\
M \ar@{-->}[rru]_\phi         & & \\
}$$
\begin{proof}
Since $M$ is a free monoid with basis $X$,
each element $x\in M$ has a unique representation $x=x_1x_2\ldots x_n$, where $x_i\in X$.
Thus the mapping $\phi_0:X\rightarrow M'$ can be extended to a mapping $\phi:M\rightarrow M'$ such that
$$\phi(x_1x_2\ldots x_n) = \phi_0(x_1)\phi_0(x_2)\ldots \phi_0(x_n), \quad x_i\in X.$$
The uniqueness property ensures that $\phi$ is a mapping.
This mapping is a homomorphism, since for all $x=x_1\ldots x_m\in\Sigma^*$ and $y=y_1\ldots y_n\in\Sigma^*$ with $x_i,y_j\in\Sigma$, we have
\begin{eqnarray*}
\phi(xy) 
&=& \phi(x_1\ldots x_my_1\ldots y_n) \\
&=& \phi_0(x_1)\ldots\phi_0(x_m)\phi_0(y_1)\ldots \phi_0(y_n)\\
&=& \phi(x_1\ldots x_m)\phi(y_1\ldots y_n) = \phi(x)\phi(y).
\end{eqnarray*}
Let $\psi:M\rightarrow M'$ be another homomorphism extending $\phi_0$.
Then for each $x=x_1\ldots x_n\in\Sigma^*$ with $x_i\in\Sigma$,
\begin{eqnarray*}
\phi(x_1x_2\ldots x_n) 
&=& \phi_0(x_1)\phi_0(x_2)\ldots \phi_0(x_n)\\
&=& \psi(x_1)\psi(x_2)\ldots \psi(x_n)\\ %\quad \psi_{|X} = \phi_0,
&=& \psi(x_1x_2\ldots x_n). %\quad \psi \mbox{ Hom.}
\end{eqnarray*}
Hence, $\phi=\psi$.
\end{proof}

\begin{example}\label{e-sigma0}
\mbox{ }
\begin{itemize}
\item The trivial monoid $M=\{e\}$ is free with basis $X=\emptyset$.
\item The monoid $(\NN_0,+,0)$ is free with basis $X=\{1\}$, since each natural number $n$ can be uniquely written as $n = n\cdot 1 =1+\ldots+1$.
\item The monoid $(\NN,\cdot,1)$ is not free.
Suppose it would be free with basis $X$.
Take the word monoid $\Sigma^*$ with basis $\Sigma=X$ und the mapping $\phi_0:X\rightarrow \Sigma^*: n\mapsto n$.
By the universal property~\ref{t-up}, there is a unique homomorphism $\phi:\NN\rightarrow \Sigma^*$ extending $\phi_0$.
The basis $X$ contains at least two elements, since $\NN$ is more than the power of one number.
If say $2,3\in X$, then 
\begin{eqnarray*}
23 = \phi_0(2)\phi_0(3) = \phi(2\cdot 3) = \phi(6)= \phi(3\cdot 2) = \phi_0(3)\phi_0(2) = 32. 
\end{eqnarray*}
But 23 and 32 are different as strings.
A contradiction.
\item Each finite non-trivial monoid $M$ cannot be free.
Suppose $M$ would be free with basis $X$.
Take the word monoid $\Sigma^*$ with basis $\Sigma=X$ and the mapping $\phi_0:X\rightarrow\Sigma^*:x\mapsto x$.
By the universal property~\ref{t-up}, there is a unique homomorphism $\phi:M\rightarrow \Sigma^*$ extending $\phi_0$.
Since $M$ is non-trivial, there is an element $x\ne e$ in $M$.
Consider the powers $x,x^2,x^3,\ldots$.
Since $M$ is finite, there are numbers $i,j$ such that $i<j$ and $x^i=x^j$.
Then $x^i = \phi(x^i) =\phi(x^j) = x^j$.
But $x^i$ and $x^j$ are distinct as strings in $\Sigma^*$. 
A contradiction.
\end{itemize}
\end{example}

\begin{proposition}\label{p-XX}
Let $M$ and $M'$ be free monoids with finite bases $X$ and $X'$, respectively.
If $|X|=|X'|$, then $M$ and $M'$ are isomorphic.
\end{proposition}
\begin{proof}
Since $X$ and $X'$ have the same finite cardinality, we can choose bijections
$\phi_0:X\rightarrow X'$ and $\psi_0:X'\rightarrow X$ which are inverse to each other. 
By the universal property~\ref{t-up}, there are homomorphisms $\phi:M\rightarrow M'$ and $\psi:M'\rightarrow M$ extending $\phi_0$ and $\psi_0$, respectively.
Then for each $x=x_1\ldots x_m\in M$ with $x_i\in X$, 
\begin{eqnarray*}
(\psi\phi)(x) 
&=& (\psi\phi)(x_1\ldots x_m) \\ &=& \psi (\phi_0(x_1)\ldots\phi_0(x_m)) \\ &=& \psi_0(\phi_0(x_1))\ldots \psi_0(\phi_0(x_m)) \\ &=& x_1\ldots x_m=x.
\end{eqnarray*}
Thus $\psi\phi = \id_M$ and similarly $\phi\psi = \id_{M'}$.
Therefore, $\phi$ and $\psi$ are isomorphisms, as required.
\end{proof}
By the Propositions~\ref{p-sigma} and~\ref{p-XX}, and Example~\ref{e-sigma0}, we obtain the following result.
\begin{theorem}
For each number $n\geq 0$, there is exactly one free monoid (up to isomorphism) with a basis of $n$ elements.
\end{theorem}

\section{Semiautomata}

Semiautomata are deterministic finite state machines with inputs but no outputs.
They can serve as basic abstractions of various physical devices.

A {\em semiautomaton}\index{semiautomaton} is a triple $A = (S,\Sigma,\delta)$, 
where $S$ is a finite set of {\em states}\index{state}, 
$\Sigma$ is an alphabet of {\em inputs}\index{input}, 
and $\delta:S\times \Sigma\rightarrow S$ is a mapping called {\em transition function}\index{transition function}.

A semiautomaton is a deterministic finite state machine that is exactly in one state at a time.
If the semiautomaton is in state $s$ and reads the input symbol $a$, it transits into the state $s'=\delta(s,a)$.

\begin{lemma}
Let\/ $A = (S,\Sigma,\delta)$ be a semiautomaton.
Then there is a unique mapping\/ $\delta^*:S\times \Sigma^*\rightarrow S$ extending\/ $\delta$ as follows:
\begin{itemize}
\item 
$\delta^*(s,\epsilon)=s$ for all\/ $s\in S$.
\item 
$\delta^*(s,a)=\delta(s,a)$ for all\/ $s\in S$ and\/ $a\in\Sigma$.
\item 
$\delta^*(s,ax)=\delta^*(\delta(s,a),x)$ for all\/ $s\in S$, $a\in\Sigma$, and\/ $x\in\Sigma^*$.
\end{itemize}
We have\/ $\delta^*(s,xy)=\delta^*(\delta^*(s,x),y)$ for all\/ $s\in S$ and\/ $x,y\in\Sigma^*$. 
\end{lemma}
The proof is based on the universal property~\ref{t-up} and makes use of induction.

\begin{proposition}
Let $A = (S,\Sigma,\delta)$ be a semiautomaton.
Then there is an action of\/ $\Sigma^*$ on $S$ defined by
\begin{eqnarray}
s\cdot x = \delta^*(s,x), \quad s\in S,x\in\Sigma^*.
\end{eqnarray}
\end{proposition}
\begin{proof}
The operation is associative, since for all $s\in S$ and $x,y\in\Sigma^*$,
$$s\cdot (xy) = \delta^*(s,xy) = \delta^*(\delta^*(s,x),y) = (s\cdot x)\cdot y,$$ 
and $s\cdot \epsilon =\delta^*(s,\epsilon) = s$.
\end{proof}

Let $A = (S,\Sigma,\delta)$ be a semiautomaton.
For each string $x\in\Sigma^*$, define the mapping 
\begin{eqnarray}
\tau_x:S\rightarrow S: s\mapsto s\cdot x
\end{eqnarray}
which describes the effect that the string $x$ has on the set of states.

\begin{proposition}
Let $A = (S,\Sigma,\delta)$ be a semiautomaton.
The mapping $\tau:\Sigma^*\rightarrow T(S):x\mapsto \tau_x$ is a monoid homomorphism.
The image 
\begin{eqnarray}
T(A) = \{\tau_x\mid x\in\Sigma^*\}
\end{eqnarray}
is a submonoid of the full transformation monoid\/ $T(S)$ generated by the set $\{\tau_a\mid a\in \Sigma\}$.
\end{proposition}
\begin{proof}
The mapping $\tau$ is a homomorphism, since for all $x,y\in\Sigma^*$ and $s\in S$,
$$s\tau_{xy} = s\cdot(xy) = (s\cdot x)\cdot y = (s\tau_x)\tau_y,$$ i.e., $\tau_{xy} = \tau_x\tau_y$,
and $s\tau_\epsilon = s\cdot \epsilon = s$, i.e., $\tau_\epsilon = \id_S$.
By Lemma.~\ref{l-comp}, $T(A)$ is a submonoid of $T(S)$.
For each symbol $a\in\Sigma$ and each string $x\in\Sigma^*$, we have $\tau_{ax}=\tau_a\tau_x$.
Thus by induction on the length of strings, the monoid $T(A)$ is generated by the set of transformations 
$\{\tau_a\mid a\in\Sigma\}$.
\end{proof}
Therefore, the behavior of a semiautomaton $A$ given by the transition function~$\delta$ is fully described by the set of transformations 
corresponding to the input alphabet $\{\tau_a\mid a\in \Sigma\}$.
The monoid $T(A)$ is called the {\em transition monoid\/}\index{transition monoid} of the semiautomaton $A$.

\begin{example}\label{e-sa0}
Consider the semiautomaton
$A = (S,\Sigma,\delta)$ 
with state set $S=\{1,2,3\}$, input alphabet $\Sigma=\{a,b\}$, and transition function $\delta$ given by the automaton graph in Figure~\ref{f-sa0}.
The corresponding transformation monoid is generated by the transformations
$$\tau_{a} = \left( \begin{array}{ccc} 1 & 2 & 3 \\ 1 & 1 & 1\end{array} \right)
\quad\mbox{and}\quad
\tau_{b} = \left( \begin{array}{ccc} 1 & 2 & 3 \\ 2 & 2 & 3\end{array} \right).$$
We have 
$$\begin{array}{ll}
\tau_{aa} = \left( \begin{array}{ccc} 1 & 2 & 3\\ 1 & 1& 1\end{array} \right), & 
\tau_{ab} = \left( \begin{array}{ccc} 1 & 2 & 3\\ 2 & 2 & 2\end{array} \right),\\
\tau_{ba} = \left( \begin{array}{ccc} 1 & 2 & 3\\ 1 & 1 & 1\end{array} \right), & 
\tau_{bb} = \left( \begin{array}{ccc} 1 & 2 & 3\\ 2 & 2 & 3\end{array} \right).
\end{array}$$
Therefore, 
the transformation monoid $T(A)$ is given by the set of transformations $\{\id_S,\tau_a,\tau_b,\tau_{ab}\}$.
\end{example}

\begin{figure}[hbt]
\begin{center}
\setlength{\unitlength}{1.0mm}
\begin{picture}(50,20)
\setlength{\unitlength}{1mm}
\put(25,10){\makebox(0,0)[c]{
\mbox{$
\xymatrix{
%\txt{start}\ar@{-->}[d] && \\
*++[o][F-]{1} 
 \ar@(ul,dl)[]_{a}
\ar@/^/[rr]^b
%\ar@{-->}[d] 
&&
*++[o][F-]{2} 
 \ar@(ur,dr)[]^{b}
\ar@/^/[ll]^a \\
& *++[o][F-]{3} 
 \ar@(ur,dr)[]^{b}
\ar@{->}[ul]^a
%\ar@/^/[ul]^a
& \\
%\txt{end} &&
}
$}
}}
\end{picture}
\end{center}
\caption{Semiautomaton $A$.}\label{f-sa0}
\end{figure}
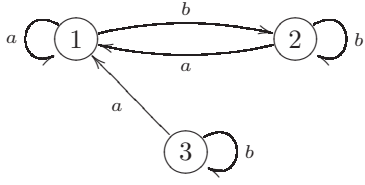

%page 45
A semiautomaton $B=(S^B,\Sigma^B,\delta^B)$ is a {\em homomorphic image\/}\index{semiautomata!homomorphic image} of a semiautomaton $A=(S^A,\Sigma^A,\delta^A)$
if there are surjective mappings $\phi:S^A\rightarrow S^B$ and $\xi:\Sigma^A\rightarrow \Sigma^B$ such that for each symbol $a\in\Sigma^A$,
\begin{eqnarray}
\tau_a^A\phi = \phi\tau_{\xi(a)}^B. 
\end{eqnarray}
%The mapping $\phi$ is then a {\em homomorphism}\index{semiautomata!homomorphism} from $A$ to $B$.
If both mappings $\phi$ and $\xi$ are also injective, the semiautomata $A$ and $B$ are said to be {\em isomorphic}.
If two semiautomata are isomorphic, one can be obtained from the other by simply renaming the states and inputs.
%page 50
\begin{proposition}
If a semiautomaton\/ $B$ is a homomorphic image of a semiautomaton\/ $A$, then\/ $T(B)$ is a homomorphic image of\/ $T(A)$.
\end{proposition}
\begin{proof}
Let $B$ be a homomorphic image of $A$.
Then there are surjective mappings $\phi:S^A\rightarrow S^B$ and $\xi:\Sigma^A\rightarrow \Sigma^B$ such that 
$\tau_a^A\phi = \phi\tau_{\xi(a)}^B$ for each $a\in\Sigma^A$.
The input mapping $\xi$ can be extended to strings by setting $\xi(a_1\ldots a_n) = \xi(a_1)\ldots\xi(a_n)$ for all $a_1,\ldots,a_n\in\Sigma^A$.
Then we obtain $\tau_x^A\phi = \phi\tau_{\xi(x)}^B$ for each string $x$ over~$\Sigma^A$.
Since $\phi$ is surjective, we have $\phi^{-1}\phi = \id_{S^B}$ and therefore for each string $x$ over~$\Sigma^A$,
$$\phi^{-1}\tau_x^A\phi = \phi^{-1}\phi\tau_{\xi(x)}^B = \tau_{\xi(x)}^B.$$ 

Consider the mapping $\psi:T(A)\rightarrow T(B)$ defined by $\tau_x^A\mapsto\tau_{\xi(x)}^B$.
This mapping is well-defined, since if $\tau_x^A=\tau_y^A$ for some strings $x,y$ over $\Sigma^A$,
then $\phi^{-1}\tau_x^A\phi = \phi^{-1}\tau_y^A\phi$ and so $\tau_{\xi(x)}^B=\tau_{\xi(x)}^B$.
This mapping is a homomorphism, since for all strings $x,y$ over $\Sigma^A$, 
$$\psi(\tau_x^A\tau_y^A) = \psi(\tau_{xy}^A) 
= \tau_{\xi(xy)}^B = \tau_{\xi(x)\xi(y)}^B 
= \tau_{\xi(x)}^B\tau_{\xi(y)}^B =\psi(\tau_x^A)\psi(\tau_y^A).$$
The mapping $\psi$ is surjective, since the mapping $\xi$ is onto and so each transformation $\tau^B\in T(B)$ 
has the form $\tau^B = \tau_{\xi(x)}^B$ for some string $x$ over~$\Sigma^A$.
\end{proof}
%page 51
\begin{example}\label{e-sa1}
Consider the semiautomata $A$ and $B$ given in Fig.~\ref{f-sa1}.
The semiautomaton $B$ is a homomorphic image of the semiautomaton $A$ with mappings
$\phi = \left( \begin{array}{ccc} 1 & 2 & 3\\ 1 & 1 & 3\end{array} \right)$ and
$\xi = \left( \begin{array}{c} a \\ a \end{array} \right)$.
The transformation monoids $T(A) = \{\id,\tau_a^A\}$ and $T(B) = \{\id,\tau_a^B\}$ with 
$\tau_a^A= \left( \begin{array}{ccc} 1 & 2 & 3\\ 3 & 3 & 3\end{array} \right)$ and
$\tau_a^B= \left( \begin{array}{cc} 1 & 3 \\ 3 & 3 \end{array} \right)$ are isomorphic,
although the semiautomata $A$ and $B$ are not isomorphic.
%the mapping $\phi$ is not one-to-one.
\end{example}
%page 51
\begin{figure}[hbt]
\begin{center}
\setlength{\unitlength}{1.0mm}
\begin{picture}(50,20)
\setlength{\unitlength}{1mm}
\put(25,10){\makebox(0,0)[c]{
\mbox{$
\xymatrix{
*++[o][F-]{1} 
 %\ar@(ul,dl)[]_{a}
\ar@{->}[rr]^a
%\ar@/^/[rr]^a
&&
*++[o][F-]{3} 
 \ar@(ur,dr)[]^{a}\\
%\ar@/^/[ll]^a \\
*++[o][F-]{2} 
% \ar@(ur,dr)[]^{b}
%\ar@/^/[urr]^a
\ar@{->}[urr]^a
&& \\
}
$}
}}
\end{picture}
\quad
\begin{picture}(50,20)
\setlength{\unitlength}{1mm}
\put(25,10){\makebox(0,0)[c]{
\mbox{$
\xymatrix{
*++[o][F-]{1} 
 %\ar@(ul,dl)[]_{a}
\ar@{->}[rr]^a
%\ar@/^/[rr]^a
&&
*++[o][F-]{3} 
 \ar@(ur,dr)[]^{a}\\
%%\ar@/^/[ll]^a \\
%*++[o][F-]{2} 
%% \ar@(ur,dr)[]^{b}
%\ar@/^/[urr]^a
%&& \\
}
$}
}}
\end{picture}
\end{center}
\caption{Semiautomata $A$ and $B$.}\label{f-sa1}
\end{figure}
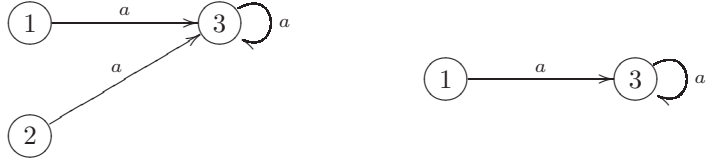

%page 44
A semiautomaton $B= (S^B,\Sigma^B,\delta^B)$ is a {\em subsemiautomaton}\index{subsemiautomaton} of a semiautomaton $A=(S^A,\Sigma^A,\delta^A)$
if $S^B\subseteq S^A$, $\Sigma^B\subseteq \Sigma^A$, and $\delta^B$ is the restriction of $\delta^A$ onto $S^B\times \Sigma^B$;
i.e., $\delta^B(s,a) = \delta^A(s,a)$ for all $s\in S^B$ and $a\in\Sigma^B$.
%page 44, Lemma
\begin{proposition}
If\/ $B$ is a subsemiautomaton of a semiautomaton\/ $A$, then\/ $T(B)$ is a homomorphic image of a subsemigroup of\/ $T(A)$.
In particular, if~$\Sigma^A=\Sigma^B$, then\/ $T(B)$ is a homomorphic image of\/ $T(A)$.
\end{proposition}
\begin{proof}
The transformations $\tau_a^A$ with $a\in\Sigma^B$ generate a subsemigroup $T$ of~$T(A)$.
The restrictions of the elements of $T$ to $\Sigma^B$ are mappings from $\Sigma^B$ to $\Sigma^B$ which form the submonoid $T(B)$.
The mapping $\phi:T\rightarrow T(B)$ which assigns to each transformation $\tau\in T$ its restriction to $S^B$ is a homomorphism.
%since $(gh)\phi = (g\phi)(h\phi)$ for all $g,h\in G$.
\end{proof}

\begin{example}\label{e-sa2}
Reconsider the semiautomaton $A$ in Figure~\ref{f-sa0}.
Take the semiautomaton
$B = (S,\Sigma,\delta)$ with state set $S=\{1,2\}$, input alphabet $\Sigma=\{a,b\}$, 
and transition function $\delta$ given by the automaton graph in Figure~\ref{f-sa2}.
The corresponding transformation monoid is generated by the transformations
$$\tau_{a} = \left( \begin{array}{cc} 1 & 2 \\ 1 & 1\end{array} \right)
\quad\mbox{and}\quad
\tau_{b} = \left( \begin{array}{cc} 1 & 2 \\ 2 & 2\end{array} \right).$$
%The semiautomaton $B$ is a subsemiautomaton of the semiautomaton $A$ in Ex.~\ref{e-sa0}.
We have 
$$\tau_{aa} = \left( \begin{array}{cc} 1 & 2 \\ 1 & 1\end{array} \right),\;
\tau_{ab} = \left( \begin{array}{cc} 1 & 2 \\ 2 & 2\end{array} \right),\;
\tau_{ba} = \left( \begin{array}{cc} 1 & 2 \\ 1 & 1\end{array} \right),
\;\mbox{and}\;
\tau_{bb} = \left( \begin{array}{cc} 1 & 2 \\ 2 & 2\end{array} \right).$$
Thus the transformation monoid $T(B)$ is given by the set of transformations $\{\id_S,\tau_a,\tau_b\}$.
The mapping $T(A)\rightarrow T(B)$ defined by restricting each transformation $\tau_x^A$ of $T(A)$ to $S^B$ is an epimorphism.
\end{example}
\begin{figure}[hbt]
\begin{center}
\setlength{\unitlength}{1.0mm}
\begin{picture}(50,20)
\setlength{\unitlength}{1mm}
\put(25,10){\makebox(0,0)[c]{
\mbox{$
\xymatrix{
%\txt{start}\ar@{-->}[d] && \\
*++[o][F-]{1} 
 \ar@(ul,dl)[]_{a}
\ar@/^/[rr]^b
%\ar@{-->}[d] 
&&
*++[o][F-]{2} 
 \ar@(ur,dr)[]^{b}
\ar@/^/[ll]^a \\
%\txt{end} &&
}
$}
}}
\end{picture}
\end{center}
\caption{Semiautomaton.}\label{f-sa2}
\end{figure}
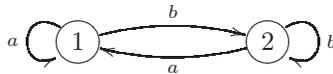

\section{Coverings of Semiautomata}
The covering of semiautomata is the most important concept in algebraic automata theory.
A covering semiautomaton is able to perform the tasks of the covered semiautomaton but may have a different structure.

%page 108
A semiautomaton $B=(S^B,\Sigma^B,\delta^B)$ {\em covers\/}\index{covering!semiautomata} 
semiautomaton $A=(S^A,\Sigma^A,\delta^A)$ written $B\geq A$, 
if there is a surjective mapping $\phi$ of a subset of $S^B$ onto $S^A$ and a mapping $\xi:\Sigma^A\rightarrow\Sigma^B$ 
such that for each symbol $a\in\Sigma^A$,
\begin{eqnarray}
\phi\tau_a^A = \tau_{\xi(a)}^B\phi. 
\end{eqnarray}
Note that if $B$ covers $A$ and the state set $S^B$ or the input alphabet $\Sigma^B$ or both are augmented to larger sets, 
the resulting semiautomaton will still cover $A$. 
%
%page 109, Lemma A
\begin{lemma}\label{l-sa-xi-inj}
If\/ $B\geq A$ such that the mapping\/ $\xi$ is injective, 
then\/ $A$ is a homomorphic image of a subsemiautomaton of\/ $B$.
\end{lemma}
\begin{proof}
Let $B$ be a covering of $A$.
Then there are appropriate mappings $\phi$ and $\xi$ such that $\phi\tau_a^A = \tau_{\xi(a)}^B\phi$ for all $a\in\Sigma^A$.

Claim that the subset $S^C = \phi^{-1}(S^A)$ of $S^B$ forms a subsemiautomaton~$C$ of~$B$ with 
input alphabet $\Sigma^C=\xi(\Sigma^A)$.
Indeed, if $s^B\in \phi^{-1}(S^A)$, then $s^B\phi\tau_a^A = s^B\tau_{\xi(a)}^B\phi$ and so $s^B\tau_{\xi(a)}^B\in\phi^{-1}(S^A)$;
i.e., the set $S^C$ is closed under the transformations in $T(B)$ which correspond to $\Sigma^C$.

Claim that $A$ is a homomorphic image of the subsemiautomaton $C$ of $B$.
%\xi inverse is used as in the def on the homomorphic image.
Indeed, we have 
$$\tau_{\xi(a)}^C\phi = \phi\tau_{\xi^{-1}(\xi(a))}^A = \phi\tau_a^A,$$ 
since $\xi$ is injective and so $\xi\xi^{-1}$ is the identity on $\Sigma^A$.
\end{proof}
%p.110, corollary
Note that if in the covering $B\geq A$, $\phi$ is injective and $\xi$ is bijective,
then the semiautomaton obtained from $A$ by coinciding the subsets of inputs having equal images under $\xi$ is isomorphic to $B$.

\begin{example}\label{e-sa3}
The semiautomaton $A$ in Fig.~\ref{f-sa0} covers the semiautomaton $B$ in Fig.~\ref{f-sa2}, where the state mapping $\phi:\{1,2\}\rightarrow\{1,2\}$ 
and the alphabet mapping $\xi:\{a,b\}\rightarrow\{a,b\}$ are the identity mappings.
\end{example}

%page 111
\begin{lemma}\label{l-sa-covtr}
Let $A$, $B$, and $C$ be semiautomata.
If\/ $C\geq B$ and\/ $B\geq A$,  then\/ $C\geq A$.
\end{lemma}
\begin{proof}
Suppose  $C$ covers $B$ and $B$ covers $A$.
Then there are surjective mappings 
$\phi_C:S_1^C\rightarrow S^B$ and
$\phi_B:S_1^B\rightarrow S^A$,
where $S_1^B\subseteq S^B$ and $S_1^C\subseteq S^C$,
and mappings
$\xi_B:\Sigma^A\rightarrow\Sigma^B$ and
$\xi_C:\Sigma^B\rightarrow\Sigma^C$
such that for all symbols $a\in\Sigma^A$ and $b\in\Sigma^B$,
$\phi_B\tau_a^A = \tau_{\xi_B(a)}\phi_B$ and
$\phi_C\tau_b^B = \tau_{\xi_C(b)}\phi_C$.
Then the composition $\phi_C\phi_B:S_1^C\rightarrow S^A$ is surjective
and the composition $\xi_B\xi_C:\Sigma^A\rightarrow\Sigma^C$ is a mapping.
For all symbols $a\in\Sigma^A$, we have
\begin{eqnarray*}
(\phi_C\phi_B)\tau_a^A 
&=& \phi_C(\phi_B\tau_a^A) 
= \phi_C(\tau_{\xi_B(a)}^B\phi_B)
= (\phi_C\tau_{\xi_B(a)}^B)\phi_B\\
&=& (\tau_{\xi_C(\xi_B(a))}^C\phi_C)\phi_B
= \tau_{\xi_C(\xi_B(a))}^C(\phi_C\phi_B).
\end{eqnarray*}
\end{proof}
The relation of covering is reflexive and transitive, but may not be symmetric.
If $A$ and $B$ are semiautomata with $A\geq B$ and $B\geq A$, then $A$ and $B$ are called {\em equivalent}\index{semiautomaton!equivalent}.

%page 112
The {\em direct product\/}\index{direct product} of semiautomata $A=(S^A,\Sigma,\delta^A)$ and $B=(S^B,\Sigma,\delta^B)$ is the semiautomaton
$A\times B = (S^{A\times B},\Sigma, \delta^{A\times B})$ with state set $S^{A\times B} = S^A\times S^B$ and transition function
\begin{eqnarray}
\delta^{A\times B} ((s^A,s^B),a) = (\delta^A(s^A,a),\delta^B(s^B,a)),\quad s^A\in S^A,s^B\in S^B,a\in\Sigma.
\end{eqnarray}
A direct product semiautomaton consists of two subsemiautomata that work in parallel and independently of each other.
\begin{proposition}
The transition monoid of a semiautomaton $A\times B$ is the direct product monoid $T(A\times B) = T(A)\times T(B)$.
\end{proposition}
\begin{proof}
Let $x\in\Sigma^*$ and $(s^A,s^B)\in S^A\times S^B$.
Then 
$$\tau_x^{A\times B}(s^A,s^B) = ({\delta^A}^*(s^A,x),{\delta^B}^*(s^B,x)) = (\tau_x^A(s_A),\tau_x^B(s_B))$$ 
and so $\tau_x^{A\times B} = (\tau_x^A,\tau_x^B)$.
Conversely, each pair $(\tau_x^A,\tau_x^B)$ gives a transformation of $T(A\times B)$.
\end{proof}

%page 113
Let $A=(S^A,\Sigma^A,\delta^A)$ and $B=(S^B,\Sigma^B,\delta^B)$ be semiautomata, 
and $\omega:S^A\times \Sigma^A\rightarrow \Sigma^B$ a mapping.
The {\em cascade product\/}\index{cascade product} of $A$ and $B$ with respect to the {\em connection mapping\/}\index{connection mapping} $\omega$ is the semiautomaton 
$$A\circ_\omega B = (S^{A\circ_\omega B},\Sigma^{A\circ_\omega B},\delta^{A\circ_\omega B})$$ 
with state set $S^{A\circ_\omega B} = S^A\times S^B$,
input alphabet $\Sigma^{A\circ_\omega B}=\Sigma^A$,
and transitions defined by
\begin{eqnarray}
(s^A,s^B)\tau_a^{A\circ_\omega B} = (s^A\tau_a^A, s^B\tau_{\omega(s^A,a)}^B), \quad s^A\in S^A,s^B\in S^B,a\in\Sigma^A.
\end{eqnarray}
The cascade product is commonly used in the case when $S^A\times \Sigma^A\subseteq \Sigma^B$ and $\omega$ is the identity mapping on $S^A\times\Sigma^A$.
Then the cascase product of $A$ and $B$ is denoted by $A\circ B$ and we have
\begin{eqnarray}
(s^A,s^B)\tau_a^{A\circ B} = (s^A\tau_a^A, s^B\tau_{(s^A,a)}^B), \quad s^A\in S^A,s^B\in S^B,a\in\Sigma^A.
\end{eqnarray}
Note that if $\Sigma^A=\Sigma^B$ and $\omega(s^A,a)=a$ for each $s^A\in S^A$ and $a\in\Sigma^A$,
then the cascade product $A\circ_\omega B$ reduces to the direct product $A\times B$.
Hence, the direct product of semiautomata is special case of the cascade product of semiautomata.
%CHECK paper ODED MALER for example cascade.

%\begin{proposition}
%The transition monoid of the cascade product\/ $A\circ_\omega B$ is the semidirect product\/ $T(A\circ_\omega ) = T(A)\rtimes_\omega T(B)$.
%\end{proposition}
%\begin{proof}
%Write $a\cdot \tau=\tau_a$.
%We have $a\cdot \tau^{A\circ_\omega B} = (a\cdot \tau^A, \omega(\cdot,a)\cdot \tau^B)$.
%The mapping $\omega$ can be extended to strings over $\Sigma^A$ by setting
%$\omega (s^A,x) $
%%Then for the transformation $\tau_x\in T(A_1\ltimes A_2)$ we have 
%%$(s_1,s_2)\tau_x 
%%= \delta^*((s_1,s_2),x) 
%%= (\delta^*_1(s_1,x), \delta^*_2(s_2,(s_1,x))) 
%%%= (s_1\tau_x^{(1)},(s_2\tau_{(s_1,x))}^{(2)})$. 
%%%and thus
%%$\tau_x = (\tau_x^{(1)},\tau_x^{(2)})$.
%\end{proof}

%page 115
%The following properties of cascade products can easily be checked.
\begin{lemma}\label{l-sa-CBA}
Let\/ $A$, $B$, and\/ $C$ be semiautomata.
If\/ $B\geq A$, then for each\/ $C\circ_\omega A$, there exists a connection mapping~$\omega'$ such that\/ 
$C\circ_{\omega'} B\geq C\circ_{\omega} A$.
\end{lemma}
\begin{proof}
Let $B\geq A$.
Then there is a mapping $\phi$ from a subset of $S^B$ onto $S^A$ and a mapping $\xi$ from $\Sigma^A$ into $\Sigma^B$ such that
for each $a\in\Sigma^A$, $\phi\tau_a^A = \tau_{\xi(a)}^B\phi$.

Define the mapping $\phi':S^C\times \phi^{-1}(S^A)\rightarrow S^C\times S^A$ by
$$(s^C,s^B)\phi' = (s^C,s^B\phi),\quad s^C\in S^C, s^B\in\phi^{-1}(S^A).$$
Moreover, define the mapping $\omega':S^C\times\Sigma^C\rightarrow\Sigma^A$ as
$$(s^C,a)\omega' = (s^C,a)\omega\xi, \quad s^C\in S^C, a\in\Sigma^C.$$
Then it follows for each $(s^C,s^B)\in S^C\times \phi^{-1}(S^A)\subseteq S^{C\circ_{\omega'}B}$ and $a\in\Sigma^C$,
$$(s^C,s^B)\phi'\tau_a^{C\circ_\omega A} = (s^C,s^B)\tau_a^{C\circ_{\omega'} B}\phi'.$$
Hence,
$C\circ_{\omega'} B\geq C\circ_{\omega} A$.
\end{proof}

%page 116
\begin{lemma}\label{l-sa-ABC}
Let\/ $A$, $B$, and\/ $C$ be semiautomata.
If\/ $B\geq A$ with $\xi$ an injective mapping, 
then for each semiautomaton $A\circ_\omega C$, there exists a connection mapping~$\omega'$ 
such that\/ $B\circ_{\omega'} C\geq A\circ_{\omega} C$.
\end{lemma}
\begin{proof}
Let $B\geq A$.
Then there is a mapping $\phi$ from a subset of $S^B$ onto $S^A$ and a mapping $\xi$ from $\Sigma^A$ into $\Sigma^B$ such that
for each $a\in\Sigma^A$, $\phi\tau_a^A = \tau_{\xi(a)}^B\phi$.
Suppose $\xi$ is injective.

Define the mapping $\phi':\phi^{-1}(S^A)\times S^C\rightarrow S^A\times S^C=S^{A\circ_\omega C}$ by
$$(s^B,s^C)\phi' = (s^B\phi,s^C),\quad s^B\in\phi^{-1}(S^A), s^C\in S^C.$$
Moreover, define the mapping $\omega':S^B\times\Sigma^B\rightarrow\Sigma^C$ such that
if $s^B\in\phi^{-1}(S^A)$ and $a\in\Sigma^A$, then 
$$(s^B,\xi(a))\omega' = (s^B\phi,a)\omega.$$
This assignment is well-defined since $\xi$ is injective.
For all other pairs $(s^B,a)$ the image under $\omega'$ can be defined arbitrarily.

Then it follows for each $(s^B,s^C)\in \phi^{-1}(S^A)\times S^C\subseteq S^{B\circ_{\omega'}C}$ and $a\in\Sigma^C$,
$$(s^B,s^C)\phi'\tau_a^{A\circ_\omega C} = (s^B,s^C)\tau_{\xi(a)}^{B\circ_{\omega'} C}\phi'.$$
Hence, $B\circ_{\omega'} C\geq A\circ_{\omega} C$.
\end{proof}

\section{Coverings with Cascade Product}
%page 98, complete means \tau_a is complete, domain = whole state set - page 98 
%page 100-101

Semiautomata coming from partitions or decompositions of state sets play an important role in 
algebraic automata theory.

Let $A=(S,\Sigma,\delta)$ be a semiautomaton.
A {\em partition\/}\index{partition} of $S$ is a collection of non-empty pairwise disjoint subsets of $S$, whose union is the whole set $S$.
The elements of a partition are called {\em blocks}\index{block}.

A partition $P$ of $S$ is called {\em admissible\/}\index{partition!admissible} if it is compatible with the monoid action of the semiautomaton; i.e.,
for each input symbol $a\in\Sigma$ and each block~$B$ in $P$, the set $\tau_a(B)=\{\tau_a(b)\mid b\in B\}$ is contained in a block of $P$.
Note that for each string $x\in\Sigma^*$ and each block $B$ of $P$, it follows that $\tau_x(B)$ is contained in a block of $P$.

A partition $P$ of a set $S$ is a {\em refinement\/}\index{refinement} of a partition $Q$ of $S$ 
if each block of $P$ is a subset of some block of $Q$.
This {\em finer}\index{finer relation} relation on the set of partitions of a set~$S$ is a partial order.
Each element has a least upper bound and a greatest lower bound.
Thus the partitions of a set $S$ form a lattice.
The coarsest partition of $S$ is $\{S\}$ consisting of $S$ as the only block 
and the finest partition is $\{\{s\}\mid s\in S\}$ consisting of the elements of $S$ as singleton blocks.
% (Fig.~\ref{f-4part}).
%\begin{figure}
%\centerline{\includegraphics[height=6cm]{set_part.eps}}
%\caption{The partitions of a four-element set ordered by refinement (Wikipedia, 2016).}\label{f-4part}
%\end{figure}

\begin{example}\label{e-sa4}
Consider the semiautomaton $A$ with state set $S= \{1,\ldots,7\}$, input alphabet $\Sigma=\{a,b\}$, 
and transitions given by 
$$\tau_a = \left( \begin{array}{ccccccc}
1 & 2 & 3 & 4 & 5 & 6 & 7\\
2 & 1 & 3 & 5 & 4 & 5 & 7
\end{array} \right)
\quad\mbox{and}\quad
\tau_b = \left( \begin{array}{ccccccc}
1 & 2 & 3 & 4 & 5 & 6 & 7\\
5 & 6 & 5 & 1 & 2 & 2 & 7
\end{array} \right).$$
Then $P = \{P_1 = \{1,2,3\}, P_2 = \{4,5,6\}, P_3=\{7\}\}$ is an admissible partition of $S$, since
$$\begin{array}{llllll}
P_1\tau_a =\{1,2,3\} \subseteq P_1, && P_2\tau_a =\{4,5\} \subseteq P_2, && P_3\tau_a=\{7\} \subseteq P_3,\\
P_1\tau_b =\{5,6\} \subseteq P_2,   && P_2\tau_b =\{1,2\} \subseteq P_1, && P_3\tau_b=\{7\} \subseteq P_3.
\end{array}$$
\end{example}

Let $A=(S^A,\Sigma^A,\delta^A)$ be a semiautomaton.
Suppose $P=\{P_1,\ldots,P_s\}$ is an admissible partition of the state set $S^A$.
Then there is a semiautomaton $B=A/P$ with state set $S^B = \{\bar P_1,\ldots,\bar P_s\}$ and input alphabet $\Sigma^B=\Sigma^A$.
Since the partition $P$ is admissible, for each symbol $a\in\Sigma^A$ and each block $P_i$ there is a block $P_j$ such that $\tau_a^A(P_i)\subseteq P_j$.
Thus the transition function $\delta^B$ of $B$ can be defined as $\delta^B(\bar P_i,a) =\bar P_j$.
The semiautomaton $B$ is called the {\em $P$-factor}\index{P-factor} of $A$.

\begin{example}\label{e-sa5}
The $P$-factor of the semiautomaton $A$ in Ex.~\ref{e-sa4} 
is the semiautomaton $B=A/P$ with state set $S^B= \{\bar P_1,\bar P_2,\bar P_3\}$,
input alphabet $\Sigma^B=\{a,b\}$, and transitions given by 
$$\tau_a^B = \left( \begin{array}{ccc} \bar P_1 & \bar P_2 & \bar P_3 \\ \bar P_1 & \bar P_2 & \bar P_3 \end{array} \right)
\quad\mbox{and}\quad
\tau_b^B = \left( \begin{array}{ccc}   \bar P_1 & \bar P_2 & \bar P_3 \\ \bar P_2 & \bar P_1 & \bar P_3 \end{array} \right).$$
\end{example}

%page 111, OBEN
\begin{proposition}
Let\/ $B$ be a semiautomaton, $P$ an admissible partition of\/ $S^B$, 
and $A=B/P$ the\/ $P$-factor of\/ $B$.
Then\/ $B\geq A$.
\end{proposition}
\begin{proof}
Let $P=\{P_1,\ldots,P_s\}$ be a partition of $S^B$.
Then define $\phi:S^B\rightarrow S^A$ by setting $\phi(s^B)=\bar P_i$ if and only if $s^B\in P_i$,
and $\xi:\Sigma^A\rightarrow\Sigma^B$ the identity mapping.
Then for each $s^B\in S^B$ with $s^B\in P_i$ and $a\in\Sigma^B$,
$s^B(\phi\tau_a^A) = \bar P_i\tau_a^A=\delta^A(\bar P_i,a)=\bar P_j$ 
if and only if $s^B(\tau_a^B\phi) = \delta^B(s^B,a)\phi =\bar P_j$.
\end{proof}

In the following, let $\pi_P:S\rightarrow P$ denote the natural mapping of $S$ onto the set of blocks of a partition $P$ of $S$;
i.e., $\pi_P(s) = P_i$ if and only if $s\in P_i$.
Moreover, let $m(P)$ denote the maximal number of elements in a block of a partition (or decomposition) $P$ of a finite set.

%page 117
\begin{theorem}\label{t-sa-part}
Let\/ $A=(S^A,\Sigma,\delta^A)$ be a semiautomaton and $P$ an admissible partition of\/ $S^A$.
Then there is a semiautomaton\/ $C$ with $|S^C|=m(P)$ such that the cascade product\/ $B\circ C$ covers $A$, where\/ $B = A/P$ is the\/ $P$-factor of\/ $A$.
\end{theorem}
\begin{proof}
Let $P=\{P_1,\ldots,P_s\}$ be an admissible partition of $S^A$, and 
$B$ the $P$-factor of $A$ with state set $S^B=\{\bar P_1,\ldots,\bar P_s\}$ and input alphabet $\Sigma$.
One can find a partition $Q = \{Q_1,\ldots,Q_t\}$ of $S^A$ with $t=m(P)$ elements such that the intersection of the partitions $P$ and $Q$,
\begin{eqnarray}
P\cap Q = \{P_i\cap Q_j\mid P_i\cap Q_j\ne \emptyset, 1\leq i\leq s,1\leq j\leq t\},
\end{eqnarray}
equals the finest partition of $S^A$; i.e., each block in $P\cap Q$ is a singleton set; i.e., $P\cap Q =  \{\{s\}\mid s\in S^A\}$.

Consider the semiautomaton $C$ with state set $S^C=\{\bar Q_1,\ldots,\bar Q_t\}$, 
input alphabet $\Sigma^C = S^B\times \Sigma$, and transitions given by
$$\bar Q_j\tau_{(\bar P_i,a)}^C = \overline{(Q_j\cap P_i)\tau_a^A\pi_Q},\quad a\in \Sigma.$$
This operation is well-defined, since $Q_j\cap P_i$ is either empty or a singleton set $\{s^A\}$.
In the first case, the transition can be chosen arbitrarily.
In the second, $\overline{s^A\tau_a^A\pi_Q}$ is the state $\bar Q_i$ of $C$ where $s^A\tau_a^A\in Q_i$.

Let $T^{B\circ C}$ denote the subset of $S^{B\circ C}$ consisting of all pairs $(\bar P_i,\bar Q_j)$ such that $P_i\cap Q_j\ne \emptyset$.
Then by construction
there is a bijective mapping $\phi: T^{B\circ C}\rightarrow S^A$ defined by $(\bar P_i,\bar Q_j)\mapsto P_i\cap Q_j$.
Thus for each element $(\bar P_i,\bar Q_j)$ in $T^{B\circ C}$ and each symbol $a\in\Sigma$, 
\begin{eqnarray*}
(\bar P_i,\bar Q_j)\phi \tau_a^A 
&=&  (P_i\cap Q_j)\tau_a^A \\
&=& P_i\tau_a^A\pi_P \cap (P_i\cap Q_j)\tau_a^A\pi_Q\\
&=& (\overline{P_i\tau_a^A\pi_P},  \overline{(P_i\cap Q_j)\tau_a^A\pi_Q})\phi\\
&=& (\bar P_i\tau_a^B,  \bar Q_j\tau_{(\bar P_i,a)}^C)\phi\\
&=& (\bar P_i,  \bar Q_j)\tau_{a}^{B\circ C}\phi.
\end{eqnarray*}
Hence, $B\circ C$ covers $A$.
\end{proof}

Note that if $Q$ is an admissible partition of $S^A$, then the mappings $\tau_{(\bar P_i,a)}^C$ are independent of $\bar P_i$,
since all elements of $Q_j$ are mapped by $\tau_a^A$ into the same block of $Q$.
Then all inputs $(\bar P_i,a)$ with the same symbol $a\in\Sigma^A$ are equal 
and after identifying them, the transitions of $C$ are of the form $\bar Q_j\tau_a^C = \overline{Q_j\tau_a^A\pi_Q}$.
Therefore, the cascade product $B\circ C$ reduces to the direct product $B\times C$.

\begin{example}\label{e-sa6}
Consider the semiautomaton $A$ in Ex.~\ref{e-sa4} 
together with the admissible partition $P$ of $S^A$ 
and the $P$-factor $B=A/P$ in Ex.~\ref{e-sa5}. 
Note that we have $m(P)=3$. 

The partition $Q = \{Q_1=\{1,4,7\},Q_2=\{2,5\},Q_3=\{3,6\}\}$ of $S^A$ has the property that 
$P\cap Q$ is the finest partition of $S^A$. %=\{\{1\},\ldots,\{7\}\}$.
Indeed, 
$P_1\cap Q_1 = \{1\}$,
$P_1\cap Q_2 = \{2\}$,
$P_1\cap Q_3 = \{3\}$,
and so on.

The semiautomaton $C$ has state set $S^C=\{\bar Q_1,\bar Q_2,\bar Q_3\}$, 
input alphabet $\Sigma^C = S^B\times \Sigma$,
and transitions given by the following table:
$$\begin{array}{c||ccc}
C            & \bar Q_1 & \bar Q_2 & \bar Q_3 \\\hline\hline
(\bar P_1,a) & \bar Q_2 & \bar Q_1 & \bar Q_3 \\
(\bar P_1,b) & \bar Q_2 & \bar Q_3 & \bar Q_2 \\
(\bar P_2,a) & \bar Q_2 & \bar Q_1 & \bar Q_2 \\
(\bar P_2,b) & \bar Q_1 & \bar Q_2 & \bar Q_2 \\
(\bar P_3,a) & \bar Q_1 & \bar Q_1^* & \bar Q_1^* \\
(\bar P_3,b) & \bar Q_1 & \bar Q_1^* & \bar Q_1^* 
\end{array}$$
For instance, $\bar Q_1\tau_{(\bar P_1,a)}^C = \overline{\{1\}\tau_a^A\pi_Q} = \overline{\{2\}\pi_Q} = \bar Q_2$.
Note that the marked (*) transitions can be chosen arbitrarily, since the intersections of blocks are empty.

The semiautomaton $B\circ C$ has the state set
$S^{B\circ C} = S^B\times S^C$,
the input alphabet
$\Sigma^{B\circ C} = \Sigma^B$,
and the transitions given by the following table:
$$\begin{array}{c||ccccc}
B\circ C &(\bar P_1,\bar Q_1) 
         &(\bar P_1,\bar Q_2)
         &(\bar P_1,\bar Q_3)
         &(\bar P_2,\bar Q_1)
         &(\bar P_2,\bar Q_2)\\\hline
a        &(\bar P_1,\bar Q_2)
         &(\bar P_1,\bar Q_1)
         &(\bar P_1,\bar Q_3)
         &(\bar P_2,\bar Q_2)
         &(\bar P_2,\bar Q_1)\\
b        &(\bar P_2,\bar Q_2)
         &(\bar P_2,\bar Q_3)
         &(\bar P_2,\bar Q_2)
         &(\bar P_1,\bar Q_1)
         &(\bar P_1,\bar Q_2)\\ \hline\hline
         &(\bar P_2,\bar Q_3) &(\bar P_3,\bar Q_1) &(\bar P_3,\bar Q_2) &(\bar P_3,\bar Q_3) \\\hline
a        &(\bar P_2,\bar Q_2)&(\bar P_3,\bar Q_1)&(\bar P_3,\bar Q_1)&(\bar P_3,\bar Q_1)\\
b        &(\bar P_1,\bar Q_2)&(\bar P_3,\bar Q_1)&(\bar P_3,\bar Q_1)&(\bar P_3,\bar Q_1)\\
\end{array}$$
%The covering $B\circ C\geq A$ is given by the state mapping
%$$\begin{array}{c||ccccccccc}
%B\circ C&(\bar P_1,\bar Q_1)&(\bar P_1,\bar Q_2)&(\bar P_1,\bar Q_3)&(\bar P_2,\bar Q_1)&(\bar P_2,\bar Q_2)&(\bar P_2,\bar Q_3)&(\bar P_3,\bar Q_1)&(\bar P_3,\bar Q_2)&(\bar P_3,\bar Q_3)
        %& 
%\end{array}$$
\end{example}

%page 101
Let $A=(S,\Sigma,\delta)$ be a semiautomaton.
A {\em decomposition\/}\index{decomposition} of $S$ is a collection of non-empty subsets of $S$, whose union is the whole set $S$.
The elements of a decomposition are called {\em blocks}\index{block}.
If the blocks of a decomposition are pairwise disjoint, the decomposition is a partition.
A decomposition $D$ of $S$ is called {\em admissible\/}\index{decomposition!admissible} 
if for each input symbol $a\in\Sigma$ and each block $B$ in $D$, there is at least one block of $D$
containing the set $\tau_a(B)=\{\tau_a(b)\mid b\in B\}$.
Note that for each string $x\in\Sigma^*$ and each block $B$ of $D$, it follows that 
there is a least one block of $D$ including $\tau_x(B)$.
For an admissible decomposition $D$ of $S$, 
the {\em $D$-factor}\index{D-factor} of $A$ is defined as for a partition of $S$.
However, if for some input symbol $a\in\Sigma$ and some block $B$ in $D$, $\tau_a(B)$ may be contained 
in more than one block of $D$ and then one of these blocks can be arbitrarily chosen.

%automata A and D-factor B from page 119/120
\begin{example}\label{e-sa7}
Consider the semiautomaton $A = (S,\Sigma,\delta)$ with state set $S=\{1,2,3,4,5,6\}$, input alphabet $\Sigma=\{a,b\}$, and transitions given by
$$\tau_a = \left( \begin{array}{cccccc}
1 & 2 & 3 & 4 & 5 & 6 \\
2 & 1 & 3 & 1 & 3 & 5 
\end{array} \right)
\quad\mbox{and}\quad
\tau_b = \left( \begin{array}{cccccc}
1 & 2 & 3 & 4 & 5 & 6 \\
5 & 4 & 3 & 3 & 3 & 3 
\end{array} \right).$$
Then $D = \{D_1 = \{1,2,3\}, D_2 = \{3,4,5\}, D_3=\{5,6\}\}$ is an admissible decompositon of $S$, since
$$\begin{array}{llllll}
D_1\tau_a =\{1,2,3\}\subseteq D_1, && D_2\tau_a =\{1,3\}\subseteq D_1, && D_3\tau_a=\{3,5\}\subseteq D_2,\\
D_1\tau_b =\{3,4,5\}\subseteq D_2, && D_2\tau_b =\{3\}\subseteq D_1, D_2, && D_3\tau_b=\{3\}\subseteq D_1, D_2.
\end{array}$$
A $D$-factor $B=A/D$ of $A$ can be defined with state set $S^D=\{\bar D_1,\bar D_2,\bar D_3\}$, input alphabet $\Sigma^D=\{a,b\}$, and transitions given by the table
$$\begin{array}{c||ccc}
B & \bar D_1 & \bar D_2 & \bar D_3 \\\hline\hline
a & \bar D_1 & \bar D_1 & \bar D_2 \\
b & \bar D_2 & \bar D_1 & \bar D_1 \\
\end{array}$$
Note that $\bar D_2\tau_b^B$ and $\bar D_3\tau_b^B$ can be defined as either $\bar D_1$ or $\bar D_2$.
\end{example}

%page 111-112
It is more convenient to work with partitions than with decompositions.
For this, the following construction due to Michael Yoeli (1963) is useful.
To this end, let $A=(S^A,\Sigma^A,\delta^A)$ be a semiautomaton, 
$D=\{D_1,\ldots,D_s\}$ be an admissible decomposition of $S^A$,
and $B=A/D$ be a $D$-factor.

Define an {\em auxiliary semiautomaton} $A^* = (S^{A^*},\Sigma^{A^*},\delta^{A^*})$ from $A$ and $B$
with state set $S^{A^*} = \{(s^A,\bar D_i)\mid s^A\in S^A\cap D_i\}$,
input alphabet $\Sigma^{A^*} = \Sigma^A = \Sigma^B$, and transitions given by
\begin{eqnarray}
(s^A,\bar D_i)\tau_a^{A^*} = (s^A\tau_a^A,\bar D_i\tau_a^B), \quad a\in\Sigma^{A^*}.
\end{eqnarray}
Note that this operation is well-defined, 
since if $s^A\in D_i$ and $\bar D_i\tau_a^B = \bar D_j$ for some $1\leq j\leq s$, then $s^A\tau_a^A\in D_j$.

Define the partition $D^*$ of $S^{A^*}$, where each block consists of all pairs~$(s^A,\bar D_i)$ which have the same second component,
i.e., $D^* = \{D^*_1,\ldots,D^*_s\}$ with $D^*_i = \{(s^A,D_i)\mid s^A\in S^A\cap D_i\}$.

Then $D^*$ is an admissible partition, since if $s^A\in D_i$ and $\bar D_i\tau_a^B = \bar D_j$ 
for some $1\leq j\leq s$, 
then $s^A\tau_a^A\in D_j$ and so $(s^A,\bar D_i)\tau_a^{A^*} = (s^A\tau_a^A,\bar D_i\tau_a^B) = (s^A\tau_a^A, \bar D_j)$;
i.e., $D_i\tau_a^{A^*} \subseteq D_j$.
%\subseteq D^*_j$, where $\bar D_i\tau_a^B=\bar D_j$.
%$D^*_i\tau_a^{A^*}\subseteq D^*_j$, where $\bar D_i\tau_a^B=\bar D_j$.

The $D^*$-factor $B^*=A^*/D^*$ of $A^*$ has state set $S^{B^*} = \{\bar D_1^*,\ldots,\bar D_s^*\}$, 
input alphabet $\Sigma^{B^*} = \Sigma^{A^*}$, and the transitions are defined as 
$\bar D_i^*\tau_a^{B^*} = \bar D_j$, where $D_i\tau_a^{A^*}\subseteq D_j$.
%page 111/112

\begin{proposition}\label{p-sa-stern}
The semiautomata\/ $B^*$ and\/ $B$ are isomorphic and moreover $A^*$ covers~$A$.
\end{proposition}
\begin{proof}
Take the bijective mapping $\phi:S^{B^*}\rightarrow S^B:\bar D_i^*\mapsto \bar D_i$ and the identity mapping $\xi:\Sigma^{B^*}\rightarrow \Sigma^B$.
Then for each symbol $a\in\Sigma^{B^*}$ and $1\leq i\leq s$, we have
$\bar D_i^* \tau_a^{B^*}\phi = \bar D_j^*\phi = \bar D_j$ if and only if
$\bar D_i^*\phi\tau_a^B = \bar D_i\tau_a^B = \bar D_j$.
Therefore, the semiautomata $B^*$ and $B$ are isomorphic. 

Take the surjective mapping $\phi:S^{A^*}\rightarrow S^A:(s^A,\bar D_i) \mapsto s^A$.
Then for each symbol $a\in\Sigma^A$ and $1\leq i\leq s$, we have
$(s^A,\bar D_i)\phi\tau_a^A = s^A\tau_a^A$ and
$(s^A,\bar D_i)\tau_a^{A^*}\phi = (s^A\tau_a^A,\bar D_i\tau_a^B)\phi = s^A\tau_a^A$.
Thus $\phi\tau_a^A = \tau_a^{A^*}\phi$ and hence $A^*$ covers $A$.
\end{proof}

%\hat A automaton - page 120
\begin{example}\label{e-sa8}
In view of the semiautomaton $A$ in Ex.~\ref{e-sa7}, the auxiliary semiautomaton $A^*$
has state set 
$$S^{A^*} = \{
(1,\bar D_1), (2,\bar D_1), (3,\bar D_1), (3,\bar D_2), (4,\bar D_2), (5,\bar D_2), (5,\bar D_3), (6,\bar D_3)
\},$$
input alphabet $\{a,b\}$, and transitions given by the table
%CHECK
$$\begin{array}{c||cccccccc}
A^* & (1,\bar D_1) & (2,\bar D_1) & (3,\bar D_1) & (3,\bar D_2) & (4,\bar D_2) & (5,\bar D_2) & (5,\bar D_3) & (6,\bar D_3)\\\hline\hline
a   & (2,\bar D_1) & (1,\bar D_1) & (3,\bar D_1) & (3,\bar D_1) & (1,\bar D_1) & (3,\bar D_1) & (3,\bar D_2) & (5,\bar D_1)\\
b   & (5,\bar D_2) & (4,\bar D_2) & (3,\bar D_2) & (3,\bar D_1) & (3,\bar D_1) & (3,\bar D_1) & (3,\bar D_1) & (3,\bar D_1)\\
\end{array}$$
The corresponding partition of the state set $S^{A^*}$ is 
$D^* = \{\bar D_1^*, \bar D_2^*, \bar D_3^*\}$,
where
\begin{eqnarray*}
\bar D_1^* &=& \{(1,\bar D_1),(2,\bar D_1),(3,\bar D_1)\},\\
\bar D_2^* &=& \{(3,\bar D_2),(4,\bar D_2),(5,\bar D_2)\},\\
\bar D_3^* &=& \{(5,\bar D_3),(6,\bar D_3)\},
\end{eqnarray*}
and the semiautomaton $B^* = A^*/D^*$ has the transition table
$$\begin{array}{c||ccc}
B^* & \bar D_1^* & \bar D_2^* & \bar D_3^* \\\hline\hline
a & \bar D_1^* & \bar D_1^* & \bar D_2^* \\
b & \bar D_2^* & \bar D_1^* & \bar D_1^* \\
\end{array}$$
\end{example}

%page 119
\begin{theorem}\label{t-sa-dec}
Let\/ $A=(S^A,\Sigma,\delta^A)$ be a semiautomaton, $D$ an admissible decomposition of\/ $S^A$, and 
$B$ the $D$-factor of\/ $A$.
Then there exist semiautomata~$B^*$ and\/ $C$ such that\/ $B^*$ is isomorphic to\/ $B$, $|S^C|=m(D)$,  
and moreover $B^*\circ C\geq A$.
\end{theorem}
\begin{proof}
Consider the auxiliary semiautomaton $A^*$ constructed from $A$ and $B$.
Let $D^*$ be the admissible partition of $S^{A^*}$ associated to $D$ and $B^* = A^*/D^*$ the corresponding factor.
Then by Thm.~\ref{t-sa-part}, there is a semiautomaton $C$ such that $B^*\circ C$ covers $A^*$ and $|S^C|=m(D^*)$.
But by Prop.~\ref{p-sa-stern}, $B^*$ and $B$ are isomorphic and $A^*$ covers $A$. 
Thus by Lemma~\ref{l-sa-covtr}, $B^*\circ C$ covers $A$.
Moreover, $m(D)=m(D^*)$. 
\end{proof}

\section{Permutation and Reset Semiautomata}
This section treats specific semiautomata which play an important role in algebraic automata theory.

Let $A$ be a semiautomaton.
A {\em permutation input\/}\index{permutation input} of $A$ is an input $a\in\Sigma^A$ such that the transformation $\tau_a^A$ is a permutation of $S^A$; 
i.e., $\tau_a^A$ is a bijective mapping.

A {\em reset input\/}\index{reset input} of $A$ is an input $a\in\Sigma^A$ such that 
the transformation yields $|S^A\tau_a^A|=1$, i.e., $S^A\tau_a^A$ is a singleton set.

A {\em permutation-reset semiautomaton\/}\index{permutation-reset semiautomaton} is a semiautomaton $A$ whose inputs $a\in\Sigma^A$ are either reset or permutation inputs.

\begin{theorem}\label{t-sa-prs}
Each semiautomaton $A$ with $n=|S^A|\geq 2$ states can be covered by a cascade product of at most $n-1$ permutation-reset semiautomata.
\end{theorem}
\begin{proof}
Take the decomposition $D$ of $S^A$, whose blocks are all the $n-1$ element subsets of $S^A$.
For each subset $S$ of $S^A$ and each input $a\in\Sigma^A$, 
we have $|S\tau_a^A|\leq |S|$.
Thus $D$ is an admissible decomposition of $S^A$.

Consider the corresponding $D$-factor $B=A/D$.
If $|S^A\tau_a^A|<n$, there is a block $D_i$ of $D$ such that $S^A\tau_a^A\subseteq D_i$. 
It follows that for each block $D_j$ of $D$,
$D_j\tau_a^A \subseteq S^A\tau_a^A \subseteq D_i$.
Thus in the $D$-factor $B$, we can define $\bar D_i\tau_a^A=\bar D_j$ for each block $D_i$.
Hence, $a$ is a reset input of $B$.

If $|S^A\tau_a^A|=n$, then $\tau_a^A$ is a permutation of $S^A$,
and
%Then for each block $D_i$ there is a block $D_j$ such that $D_i\tau_a^A=D_j$.
the permutation $\tau_a^A$ provides a permutation of the blocks of $D$.
Thus the transition $\tau_a^B$ is also a permutation of $S^B$ and hence $a$ is a permutation input of $B$.

It follows that the $D$-factor $B$ is a permutation-reset semiautomaton.
By Thm.~\ref{t-sa-dec}, the semiautomaton $A$ can be covered by a cascade product $B\circ C$, where
$B$ is a permutation-reset semiautomaton and $C$ is a semiautomaton with $m(D)=n-1$ states.
Applying the same construction to the semiautomaton $C$ provides a covering of $C$ by a cascade product $E\circ F$, where
$E$ is a permutation-reset semiautomaton and $F$ is a semiautomaton with $n-2$ states.
By Lemma~\ref{l-sa-CBA}, $B\circ (E\circ F)\geq B\circ C$ and so by Lemma~\ref{l-sa-covtr}, $B\circ (E\circ F)\geq A$.
Since each two-state semiautomaton is a permutation-reset semiautomaton, the repetition of this construction yields the desired result.
\end{proof}

\begin{example}\label{e-sa-pra}
Let $A$ be a semiautomaton with state set $S^A=\{1,2,3,4,5\}$, input alphabet $\Sigma^A=\{a,b\}$, and transitions
$$\tau_a^A = \left(\begin{array}{ccccc} 1 & 2 & 3 & 4 & 5\\ 2 & 3 & 4 & 5 & 1 \end{array} \right) 
\quad \mbox{and}\quad
\tau_b^A = \left(\begin{array}{ccccc} 1 & 2 & 3 & 4 & 5\\ 2 & 4 & 5 & 3 & 2 \end{array} \right).$$
Consider the decomposition $D$ of the state set $S^A$ into the $|S^A|-1$ element blocks
$$\begin{array}{lll}
D_1 = \{2,3,4,5\},&
D_2 = \{1,3,4,5\},&
D_3 = \{1,2,4,5\},\\
D_4 = \{1,2,3,5\},&
D_5 = \{1,2,3,4\}.
\end{array}$$
This decomposition is admissible, and the corresponding $D$-factor $B=A/D$ has state set $S^B=\{\bar D_1,\bar D_2,\bar D_3,\bar D_4,\bar D_5\}$,
input alphabet $\Sigma^B=\{a,b\}$, and transitions
$$\tau_a^B = \left(\begin{array}{ccccc} 
\bar D_1,\bar D_2,\bar D_3,\bar D_4,\bar D_5\\
\bar D_2,\bar D_3,\bar D_4,\bar D_5,\bar D_1
\end{array} \right) 
\quad \mbox{and}\quad
\tau_b^B = \left(\begin{array}{ccccc} 
\bar D_1,\bar D_2,\bar D_3,\bar D_4,\bar D_5\\
\bar D_1,\bar D_1,\bar D_1,\bar D_1,\bar D_1
\end{array} \right). $$
Thus $a$ is permutation input and $b$ is a reset input of $B$.
Hence, $B$ is a permutation-reset semiautomaton.
\end{example}

A {\em permutation semiautomaton}\index{permutation semiautomaton} is a semiautomaton $\Pi$ where each 
transformation $\tau_a^\Pi$ with $a\in\Sigma^\Pi$ is a permutation of $S^\Pi$.

A {\em reset semiautomaton}\index{reset semiautomaton} is a semiautomaton $R$ where each transformation $\tau_a^R$
with $a\in\Sigma^R$ is either the identity mapping or a reset mapping.

%p. 124 ff
\begin{theorem}\label{t-sa-pr}
Each permutation-reset semiautomaton $A$ can be covered by a cascade product $\Pi\circ R$ of a permutation semiautomaton $\Pi$ and a reset semiautomaton $R$.
\end{theorem}
\begin{proof}
Decompose the state set $\Sigma=\Sigma^A$ of $A$ into two disjoint subsets $\Sigma = \Sigma_p\cup\Sigma_r$,
where $\Sigma_p$ and $\Sigma_r$ are the sets of permutation and reset inputs of $A$, respectively.

Take the transformation monoid $T(\Pi)$ generated (as group) 
by the permutations $\tau_a^A$ corresponding to the permutation inputs $a\in\Sigma_p$.
Clearly, $T(\Pi)$ is a subgroup of $T(A)$.
The elements of $T(\Pi)$ are the permutations $\tau_x^A$, where $x\in\Sigma_p^*$.
They form the state set of a semiautomaton $\Pi$ which are denoted by $\overline{x_p^A}$.
The semiautomaton $\Pi$ has the input alphabet $\Sigma^\Pi=\Sigma^A$ and the transitions
$$\overline{x_p^A}\tau_{a_p}^\Pi = \overline{(xa_p)_p^A}\quad \mbox{and}\quad
\overline{x_p^A}\tau_{a_r}^\Pi = \overline{x_p^A},
\quad a_p\in\Sigma_p, a_r\in\Sigma_r, x_p\in\Sigma_p^*.$$
Thus $\Pi$ is a permutation semiautomaton.

Define the semiautomaton $R$ with state set $S^R = \{\overline{s^A}\mid s^A\in S^A\}$, input alphabet $\Sigma^R = S^\Pi\times \Sigma$, and transitions
$$\overline{s^A}\tau_{(\overline{x_p^A},a_p)}^R = \overline{s^A}
\quad\mbox{and}\quad
\overline{s^A}\tau_{(\overline{x_p^A},a_r)}^R = \overline{(s^A\tau_{a_r}^A)(\tau_x^A)^{-1}},
\; a_p\in\Sigma_p, a_r\in\Sigma_r, x_p\in\Sigma_p^*.$$
Since $s^A\tau_{a_r}^A$ is the same for all $s^A$, it follows that $R$ is a reset semiautomaton.

Claim that the cascade product $\Pi\circ R$ covers $A$.
Indeed,
define the mapping $\phi:S^{\Pi\circ R}\rightarrow S^A$ by setting
$$\phi(\overline{x_p^A},\overline{s^A}) = s^A\tau_x^A.$$
Since $\tau_x^A$ is a permutation of $S^A$, the mapping $\phi$ is surjective.
Moreover, since the cascade product 
$\Pi\circ R$ has the input alphabet $\Sigma$, the mapping $\xi:\Sigma^{\Pi\circ R}\rightarrow \Sigma^A$ can be chosen as the identity.
The equation $\phi\tau_a^A = \tau_a^{\Pi\circ R}\phi$ with $a\in\Sigma$ can be easily checked.
\end{proof}

\begin{example}\label{e-sa-pra1}
Consider the semiautomaton $A$ with state set $S^A = \{1,2,3,4,5\}$, input alphabet $\Sigma^A=\{a,b\}$, and transitions
$$\tau_a^A = \left(\begin{array}{ccccc} 1 & 2 & 3 & 4 & 5\\ 2 & 3 & 4 & 5 & 1 \end{array} \right) 
\quad \mbox{and}\quad
\tau_b^A = \left(\begin{array}{ccccc} 1 & 2 & 3 & 4 & 5\\ 1 & 1 & 1 & 1 & 2 \end{array} \right).$$
This is a permutation-reset semiautomaton with permutation input $a$ and reset input $b$. 
The state set $\Sigma$ decomposes into the set of permutation inputs $\Sigma_p=\{a\}$ and the set of reset inputs $\Sigma_r=\{b\}$
(see semiautomaton $B$ in Ex.~\ref{e-sa-pra}).

The semiautomaton $A$ is covered by a cascade product $\Pi\circ R$ with permutation semiautomaton $\Pi$ and reset semiautomaton $R$.
The permutation semiautomaton $\Pi$ has the transformation monoid $T(\Pi)$ generated by the permutation $g=\tau_a^A$ of the permutation input $a\in\Sigma_p$.
This is a cyclic group of order~5,
$$T(\Pi) = \{g=\tau_a^A,g^2=\tau_{aa}^A,g^3=\tau_{aaa}^A,g^4=\tau_{aaaa}^A,\id=g^5=\tau_{aaaaa}^A\}.$$ 
Thus the semiautomaton $\Pi$ has state set $S^\Pi = \{\overline{\id},\overline g,\ldots,\overline {g^4}\}$, 
input alphabet~$\Sigma$, and the transitions 
$$\overline{g^i}\tau_a^\Pi = \overline{g^{i+1}}
\quad \mbox{and} \quad
\overline{g^i}\tau_b^\Pi = \overline{g^i}, \quad 0\leq i\leq 4,$$
where $i+1$ is taken modulo~5.

The reset semiautomaton $R$ has state set $\Sigma^R=\{\overline {1^A},\ldots,\overline {5^A}\}$, input alphabet $\Sigma^R=S^\Pi\times \Sigma$, and transitions 
$$\overline {j^A} \tau_{(\overline{g^i},a)}= \overline{j^A},\quad 0\leq i\leq 4, 1\leq j\leq 5 ,$$
and
$$\begin{array}{c||ccccc}
R                  & \overline {1^A}  & \overline {2^A}  & \overline {3^A}  & \overline {4^A}  & \overline {5^A}  \\\hline 
(\overline{\id},b) & \overline {1^A}  & \overline {1^A}  & \overline {1^A}  & \overline {1^A}  & \overline {1^A}  \\
(\overline{g},b)   & \overline {5^A}  & \overline {5^A}  & \overline {5^A}  & \overline {5^A}  & \overline {5^A}  \\
(\overline{g^2},b) & \overline {4^A}  & \overline {4^A}  & \overline {4^A}  & \overline {4^A}  & \overline {4^A}  \\
(\overline{g^3},b) & \overline {3^A}  & \overline {3^A}  & \overline {3^A}  & \overline {3^A}  & \overline {3^A}  \\
(\overline{g^4},b) & \overline {2^A}  & \overline {2^A}  & \overline {2^A}  & \overline {2^A}  & \overline {2^A}  
\end{array}$$
%\overline {j^A} \tau_{(\overline{g^i},b)}
%\mbox{ depends only on $i$},
%%= \overline{(6-i)^B},
%\quad 0\leq i\leq 4, 1\leq j\leq 5 .$$
%Note that $\sigma_{a^i}^B(1) = 1+i$ modulo  maps $1$ to 
\end{example}

%p. 126
\begin{theorem}
Each reset semiautomaton can be covered by a direct product of two-state reset semiautomata.
\end{theorem}
\begin{proof}
Let $A$ be a reset semiautomaton with $n=|S^A|\geq 2$ states.
If $n=2$, the assertion is clear. 
Otherwise note that each partition $P$ of $S^A$ is admissible, since every transition $\tau_a^A$ is either the identity mapping or a reset mapping.
Thus we may take a partition $P$ of $S^A$ with two blocks.
Then by Thm.~\ref{t-sa-part}, $A$ can be covered by a cascade product $B\circ C$, where $|S^B|=2$ and $|S^C|=m(P)<|S^A|$.
In this case, the cascade product $B\circ C$ can be chosen as the direct product~$B\times C$.

The same procedure can be applied to the semiautomaton $C$ providing a covering of $C$ by a direct product $D\times E$.
By Lemma~\ref{l-sa-CBA}, $B\times(D\times E)$ covers $B\times C$ and thus by Lemma~\ref{l-sa-covtr}, $B\times (D\times E)$ covers $A$.
Continuing this way leads to the desired result.
For this, note that a reset semiautomaton $B$ covering a semiautomaton $A$ can be augmented by one state and still covers $A$. 
\end{proof}

%p. 127
A finite group $G$ gives rise to a {\em grouplike semiautomaton}\index{grouplike semiautomaton},  
which is a semiautomaton denoted by $G$ with state set $G$, input alphabet $G$, and transitions given by the right multiplications 
$\tau_g^G:G\rightarrow G:x\mapsto xg$, where $g\in G$.

\begin{example}\label{e-sa-g0}
Consider the Klein four-group $G=\{e,a,b,c\}$ given by the group table
$$\begin{array}{c||cccc}
\cdot & e & a & b & c\\\hline\hline
e     & e & a & b & c\\
a     & a & e & c & b\\
b     & b & c & e & a\\
c     & c & b & a & e\\
\end{array}$$
The grouplike semiautomaton $G$ has state set $G$, input alphabet $G$, and the transitions are the right multiplications given by the rows of the group table,
$$
\begin{array}{ll}
\tau_e^G=\left(\begin{array}{cccc} e & a & b & c\\ e & a & b & c\\ \end{array}\right),&
\tau_a^G=\left(\begin{array}{cccc} e & a & b & c\\ a & e & c & b\\ \end{array}\right),\\
\tau_b^G=\left(\begin{array}{cccc} e & a & b & c\\ b & c & e & a\\ \end{array}\right),&
\tau_c^G=\left(\begin{array}{cccc} e & a & b & c\\ c & b & a & e\\ \end{array}\right) .
\end{array}$$
\end{example}

%page 126
\begin{lemma}
The permutation semiautomaton $\Pi$ in Thm.~\ref{t-sa-pr} can be covered by a grouplike semiautomaton.
\end{lemma}
\begin{proof}
The permutation automaton $\Pi$ has state set $G=T(\Pi)$ and input alphabet $\Sigma^\Pi=\Sigma$.
Each transition $\tau_{a_p}^\Pi$ corresponding to a permutation input $a_p\in\Sigma_p$ is a permutation which amounts to a right multiplication,
and each transition~$\tau_{a_r}^\Pi$ associated with a reset input $a_r\in\Sigma_r$ is the identity mapping.

Define the semiautomaton $G$ with state set $S^G=G$, input alphabet $\Sigma^G=G$, and transitions given by the right multiplications with the elements of $G$.
Then~$G$ covers $\Pi$ 
if we take $\phi:S^G\rightarrow S^\Pi$ as the identity mapping 
and $\xi:\Sigma^\Pi\rightarrow\Sigma^G$ as the mapping that assigns each $a\in\Sigma^\Pi$ to the element of $G$ which performs the same right multiplication as $\tau_a^\Pi$. 
\end{proof}

\begin{example}\label{e-sa-pra2}
Reconsider the permutation semiautomaton $\Pi$ in Ex.~\ref{e-sa-pra1}.
This semiautomaton is covered by the grouplike semiautomaton $G=C_5$ given by the cyclic group of order~5.
\end{example}

\begin{theorem}\label{t-sa-glike}
Let\/ $G$ be a grouplike semiautomaton and $H$ a subgroup of\/ $G$.
Then\/ $G$ can be covered by a cascade product\/ $B\circ_\omega C$ such that the semiautomaton\/ $C$ is grouplike 
and isomorphic to\/ $H$.
\end{theorem}
\begin{proof}
Let $G$ be a grouplike semiautomaton
and $H=\{e=h_1,\ldots,h_t\}$ a subgroup $G$.
Take the partition $P$ of $G$ into the right cosets of $H$:
\begin{eqnarray}\label{e-Hg}
P = \{H = He,Hg_2,\ldots,Hg_s\},
\end{eqnarray}
where $T=\{e=g_1,\ldots,g_s\}$ is a transversal\index{transveral} of $H$ in $G$, i.e., a full set of representatives of the right cosets of $H$ in $G$.
Note that by Lagrange's theorem, we have $|G|=st$.
Moreover, $P$ is an admissible partition of $G$, since the right cosets of $H$ in $G$ are pairwise disjoint and 
for each $1\leq i\leq s$ and $g\in G$, we have
$(Hg_i)\tau_g^G = (Hg_i)g = H(g_ig) = Hg_j$ where $g_ig=hg_j\in Hg_j$ for some $h\in H$ and $1\leq j\leq s$.

On the hand, the set $Q = \{T,h_2T,\ldots,h_tT\}$ is also a partition of $G$, 
since if $g\in h_iT\cap h_jT$, then $g=h_ig_k=h_jg_l$ for some $g_k,g_l\in T$.
Since the cosets in~(\ref{e-Hg}) are disjoint, $k=l$ and therefore $h_i=h_j$.
Moreover, the partition $P\cap Q$ is the finest one of $G$; i.e., each block in $P\cap Q$ is a singleton set.
To see this, let $h_ig_k = h_jg_l$. 
Then as above, $k=l$ and so $h_i=h_j$.
It follows that each group element $g\in G$ can be written as $g=h_ig_j$, where $1\leq i\leq t$ and $1\leq j\leq s$.

By Thm.~\ref{t-sa-part}, the semiautomaton $G$ can be covered by the cascade product $B\circ C$, 
where $B=G/P$ and $|S^C|=|H|=t$. 

The semiautomaton $C$ is isomorphic to the grouplike semiautomaton $H$ with state set $H$, input alphabet $H$, and transitions given by the right multiplication with 
the elements of $H$.
To see this, define the mapping $\phi:H\rightarrow S^C$ by $\phi(h_i) = \overline{h_iT}$ and 
the mapping 
$\xi:\Sigma^C\rightarrow\Sigma^H$ by $\xi(\overline{Hg_i},g) = h_j$, where $g_ig=h_jg'$ 
for some $1\leq j\leq t$ and $g'\in G$.
Both mappings are surjective.
The identity $\phi\tau_{(\overline{Hg_i},g)}^C = \tau_{\xi(\overline{Hg_i},g)}^H\phi$ can now be easily checked.

By the remark after Lemma~\ref{l-sa-xi-inj}, after identifying the inputs of $C$ giving the same image 
we obtain a semiautomaton $C'$ which is isomorphic to the grouplike semiautomaton $H$.
If the mapping $\xi$ is taken as $\omega$, the result is obtained.
\end{proof}

\begin{corollary}\label{c-sa-glike}
Let\/ $G$ be a grouplike semiautomaton and $H$ a normal subgroup of\/ $G$.
Then\/ $G$ can be covered by a cascade product\/ $B\circ_\omega C$ such that\/ $B$ is grouplike and isomorphic to\/ $G/H$ and\/ $C$ is grouplike and isomorphic to\/ $H$.
\end{corollary}
\begin{proof}
In the previous proof, 
the $P$-factor $B=G/P$ has the state set $S^B=\{\overline{He},\overline{Hg_2},\ldots,\overline{Hg_s}\}$ 
which corresponds to the set of cosets of $H$ in $G$.
By hypothesis, $H$ is normal, i.e., $gHg^{-1}=H$ for each group element $g\in G$.
Thus the cosets $H,Hg_2,\ldots,g_s$ form the factor group $G/H$.

Let $g$ and $g'$ belong to the same coset of $H$ in $G$. 
Then $g'=hg$ for some $h\in H$ and so in the $P$-factor $B$, 
$$(Hg_i)g = (g_iH)g = g_i(Hg) = g_i(Hg') = (g_iH)g' = (Hg_i)g'.$$
Thus the transitions $\tau_g^B$ and $\tau_{g'}^B$ are equal, and vice versa.
This shows that $B$ is a grouplike semiautomaton which is isomorphic to $G/H$.
\end{proof}

\begin{example}
Reconsider the grouplike semiautomaton $G$ in Ex.~\ref{e-sa-g0}.
Take the subgroup $H=\{e,a\}$ of $G$.
Then $P = \{P_1=H=\{e,a\},P_2=Hb=\{b,c\}\}$ is a partition of $G$ into right cosets
and $T=\{e,b\}$ is a transversal of $H$ in $G$.
The set $Q = \{Q_1=T = \{e,b\}, Q_2=Ta = \{a,c\}\}$ is also a partition of $G$ 
and $P\cap Q$ is the finest partition of $G$:  
$$P_1\cap Q_1=\{e\}, \;P_1\cap Q_2=\{a\}, \; P_2\cap Q_1=\{b\}, \;P_2\cap Q_2=\{c\}.$$

The semiautomaton $B=A/P$ has the state set $S^B=\{\bar P_1,\bar P_2\}$, the input alphabet $\Sigma^B=\Sigma^A$, and the transitions given by
\begin{eqnarray*}
\tau_e^B = \left(\begin{array}{cc}\bar P_1,\bar P_2\\\bar P_1,\bar P_2 \end{array}\right),&
\tau_a^B = \left(\begin{array}{cc}\bar P_1,\bar P_2\\\bar P_1,\bar P_2 \end{array}\right),\\
\tau_b^B = \left(\begin{array}{cc}\bar P_1,\bar P_2\\\bar P_2,\bar P_1 \end{array}\right),&
\tau_c^B = \left(\begin{array}{cc}\bar P_1,\bar P_2\\\bar P_2,\bar P_1 \end{array}\right).
\end{eqnarray*}
Since the group $G$ is abelian, the subgroup $H$ is normal.
Thus the semiautomaton $B$ is grouplike and isomorphic to the factor group $G/H = \{H,Hb\}$ which has the following
group table:
$$\begin{array}{c||cc}
B  & H & Hb\\\hline\hline
H  & H & Hb\\
Hb & Hb & H\\
\end{array}$$
The semiautomaton $C$ has state set $S^C=\{\bar Q_1,\bar Q_2\}$, input alphabet $\Sigma^C = S^B\times \Sigma^A$, and transitions given by following table:
$$\begin{array}{c||cc}
C            & \bar Q_1 & \bar Q_2\\\hline\hline
(\bar P_1,e) & \bar Q_1 & \bar Q_2\\
(\bar P_1,a) & \bar Q_2 & \bar Q_1\\
(\bar P_1,b) & \bar Q_1 & \bar Q_2\\
(\bar P_1,c) & \bar Q_2 & \bar Q_1\\
(\bar P_2,e) & \bar Q_2 & \bar Q_2\\
(\bar P_2,a) & \bar Q_1 & \bar Q_1\\
(\bar P_2,b) & \bar Q_2 & \bar Q_2\\
(\bar P_2,c) & \bar Q_1 & \bar Q_2\\
\end{array}$$
The grouplike semiautomaton $H$ has state set $H$, input alphabet $H$, and the transitions
$$\tau_e^H = \left(\begin{array}{cc} e & a\\e & a \end{array}\right)\quad\mbox{and}\quad
\tau_a^H = \left(\begin{array}{cc} e & a\\a & e \end{array}\right). $$
Now identify equal inputs in $C$:
$(\bar P_1,e)$, $(\bar P_1,b)$, $(\bar P_2,a)$, $(\bar P_2,c)$
and
$(\bar P_1,a)$, $(\bar P_1,c)$, $(\bar P_2,e)$, $(\bar P_2,b)$.
Then we obtain a semiautomaton $C'$ with state set $\{\bar Q_1,\bar Q_2\}$, input alphabet say
$\{(\bar P_1,e),(\bar P_1,a)\}$ and transition table
$$\begin{array}{c||cc}
C'           & \bar Q_1 & \bar Q_2\\\hline\hline
(\bar P_1,e) & \bar Q_1 & \bar Q_2\\
(\bar P_1,a) & \bar Q_2 & \bar Q_1\\
\end{array}$$
It is clear that $C'$ is isomorphic to the grouplike semiautomaton $H$.
\end{example}

%p. 129
A group $G$ is {\em simple}\index{group!simple} if it has only the trivial normal subgroups $G$ and $\{e\}$.
A {\em simple grouplike semiautomaton}\index{semiautomaton!simple grouplike} is a grouplike semiautomaton $G$ which corresponds to a simple group.

\begin{theorem}
Each grouplike semiautomaton\/ $G$ can be covered by a cascade product of simple group\-like semiautomata.
\end{theorem}
\begin{proof}
Each finite group $G$ has a {\em composition series}\index{composition series}. 
That is, a series $G=H_0,H_1,\ldots,H_r=\{e\}$ of subgroups of $G$ such that $H_{i+1}$ is a proper normal subgroup of $H_i$
and the factor groups $H_i/H_{i+1}$ are simple.

By Cor.~\ref{c-sa-glike}, the grouplike semiautomaton $G$ can be covered by a cascade product 
$G/H_1\circ_{\omega_1}H_1$, where the grouplike semiautomaton $G/H_1$ is simple.
The same construction can be applied to the grouplike semiautomaton $H_1$, which can be covered by a cascade product 
$H_1/H_2\circ_{\omega_2}H_2$, 
where the grouplike semiautomaton $H_1/H_2$ is simple.
Thus by Lemma~\ref{l-sa-CBA}, 
$$G/H\circ_{\omega'_2}(H_1/H_2\circ_{\omega_2}H_2)\geq G/H_1\circ_{\omega_1}H_1$$ 
and hence by Lemma~\ref{l-sa-covtr}, $$G/H\circ_{\omega'_2}(H_1/H_2\circ_{\omega_2}H_2)\geq G.$$
By repeating this construction the result is follows.
\end{proof}

\section{Krohn-Rhodes Decomposition}
%p. 129
%By using Lemmas~\ref{l-sa-CBA} and~\ref{l-sa-ABC},
The assertions established in the previous sections cumulate in the following result.
\begin{theorem}[Krohn-Rhodes]
Each semiautomaton\/ $A$ can be covered by direct and cascade products of semiautomata of two shapes:
\begin{itemize}
\item simple grouplike semiautomata, and
\item two-state reset semiautomata. 
\end{itemize}
\end{theorem}
%Such a covering of a semiautomaton $A$ is called a {\em Krohn-Rhodes covering}\index{Krohn-Rhodes covering} of~$A$.

This assertion can be refined by allowing only simple grouplike semiautomata with simple groups
that are homomorphic images of subgroups of the transition monoid $T(A)$.
For this, specific admissible decompositions are employed.
However, this will not be pursued any further.

Instead, we will give an example to demonstrate the construction in the above theorem.
For this, note that the construction of a covering of a given semiautomaton $A$ can always be started 
with the decomposition $D$ of $S^A$, whose blocks are all the $|S^A|-1$ element subsets of $S^A$ 
as described in the proof of Thm.~\ref{t-sa-prs}.
%This is exemplified in the following.

\begin{example}
Reconsider the semiautomaton $A$ with admissible decomposition $D=\{D_1,\ldots,D_5\}$ of $S^A$ and $D$-factor $B=A/D$ in Ex.~\ref{e-sa-pra}.

Regard the auxiliary semiautomaton $A^*$ covering $A$.
The semiautomaton $A^*$ has state set 
$S^{A^*} =\{(i,\bar D_j)\mid 1\leq i,j\leq 5,i\ne j\}$,
input alphabet
$\Sigma^{A^*} = \{a,b\}$,
and transitions
$$(i,\bar D_j) \tau_a^{A^*} = (\tau_a^A(i),\tau_a^B(\bar D_j))
\quad \mbox{and}\quad 
(i,\bar D_j) \tau_b^{A^*} = (\tau_b^A(i),\tau_b^B(\bar D_j)),$$
where $1\leq i,j\leq 5$ and $i\ne j$.
For instance, 
$(3,\bar D_2) \tau_a^{A^*} = (\tau_a^A(3),\tau_a^B(\bar D_2)) = (4,\bar D_3)$
and
$(3,\bar D_2) \tau_b^{A^*} = (\tau_b^A(3),\tau_b^B(\bar D_2)) = (5,\bar D_1)$.

Define the decomposition $D^*=\{D_1^*,\ldots,D_5^*\}$ of $S^{A^*}$ having blocks
$D_i^* = \{(j,\bar D_i)\mid 1\leq j\leq 5, i\ne j\}$ for $1\leq i\leq 5$.
This decomposition is an admissible partition.

The semiautomaton $B^* = A^*/D^*$ has the state set 
$S^{B^*} =\{\bar D_i^*\mid 1\leq i\leq 5\}$,
input alphabet
$\Sigma^{B^*} = \{a,b\}$,
and transitions given by the table
$$\begin{array}{c||ccccc}
B^*          & \bar D_1^* & \bar D_2^* & \bar D_3^* & \bar D_4^* & \bar D_5^* \\\hline\hline
\tau_a^{B^*} & \bar D_2^* & \bar D_3^* & \bar D_4^* & \bar D_5^* & \bar D_1^* \\
\tau_b^{B^*} & \bar D_1^* & \bar D_1^* & \bar D_1^* & \bar D_1^* & \bar D_1^* 
\end{array}$$
This is a permutation-reset semiautomaton isomorphic to $B$.
It is covered by the cascade product $\Pi\circ R$ of a permutation semiautomaton $\Pi$ and a reset semiautomaton $R$ (see Ex.~\ref{e-sa-pra1}).
The permutation semiautomaton $\Pi$ is covered by a grouplike semiautomaton $G$ which corresponds to the cyclic group of order~5. 
Since groups of prime order are simple, the semiautomaton $G$ is already simple.
On the other hand, the reset semiautomaton $R$ can be covered by a cascade of two-state reset semiautomata.
%cover by ... into grouplike and reset semiautomata...

In the next step, a semiautomaton $C$ can be constructed such that $B^*\circ C$ covers $A^*$.
To this end, a partition $P$ of $S^{A^*}$ is required such that $D^*\cap P$ is the finest partition of $S^{A^*}$.
For instance, we may take the partition $P=\{P_1,\ldots,P_5\}$ with blocks
\begin{eqnarray*}
P_1 &=& \{(1,\bar D_2), (1,\bar D_3), (1,\bar D_4), (1,\bar D_5)\},\\
P_2 &=& \{(2,\bar D_1), (2,\bar D_3), (2,\bar D_4), (2,\bar D_5)\},\\
P_3 &=& \{(3,\bar D_1), (3,\bar D_2), (3,\bar D_4), (3,\bar D_5)\},\\
P_4 &=& \{(4,\bar D_1), (4,\bar D_2), (4,\bar D_3), (4,\bar D_5)\},\\
P_5 &=& \{(5,\bar D_1), (5,\bar D_2), (5,\bar D_3), (5,\bar D_4)\}.
\end{eqnarray*}
Then we have $D_j^*\cap P_i = \{(i,\bar D_j)\}$ for $1\leq i,j\leq 5$ with $i\ne j$.

Define the semiautomaton $C$ with state set $S^C = \{\bar P_1,\ldots,\bar P_5\}$, 
input alphabet $\Sigma^C=S^{A^*}\times\{a,b\}$,
and transitions
$$\bar P_i\tau_{(\bar D_j^*,a)}^C = (i,\bar D_j)\tau_a^{A^*}\pi_P,
\quad\mbox{and}\quad
\bar P_i\tau_{(\bar D_j^*,b)}^C = (i,\bar D_j)\tau_b^{A^*}\pi_P,$$
where $1\leq i,j\leq 5$ with $i\ne j$.
For instance, we have 
$\bar P_3\tau_{(\bar D_2^*,a)}^C = (3,\bar D_2)\tau_a^{A^*}\pi_P = (4,\bar D_3)\pi_P = \bar P_4$ 
and
$\bar P_3\tau_{(\bar D_2^*,b)}^C = (3,\bar D_2)\tau_b^{A^*}\pi_P = (5,\bar D_1)\pi_P = \bar P_5$.
%
%$$\begin{array}{c||ccccc}
%C              & \bar P_1 & \bar P_2 & \bar P_3 & \bar P_4 & \bar P_5 \\\hline\hline
%(\bar D_1^*,a) & \bar P_2 & \bar P_2 & \bar P_2 & \bar P_2 & \bar P_2 \\
%(\bar D_2^*,a) & \bar P_3 & \bar P_3 & \bar P_3 & \bar P_3 & \bar P_3 \\
%(\bar D_3^*,a) & \bar P_4 & \bar P_4 & \bar P_4 & \bar P_4 & \bar P_4 \\
%(\bar D_4^*,a) & \bar P_5 & \bar P_5 & \bar P_5 & \bar P_5 & \bar P_5 \\
%(\bar D_5^*,a) & \bar P_1 & \bar P_1 & \bar P_1 & \bar P_1 & \bar P_1 \\
%(\bar D_1^*,b) & \bar P_1 & \bar P_1 & \bar P_1 & \bar P_1 & \bar P_1 \\
%(\bar D_2^*,b) & \bar P_1 & \bar P_1 & \bar P_1 & \bar P_1 & \bar P_1 \\
%(\bar D_3^*,b) & \bar P_1 & \bar P_2 & \bar P_3 & \bar P_4 & \bar P_5 \\
%(\bar D_4^*,b) & \bar P_1 & \bar P_2 & \bar P_3 & \bar P_4 & \bar P_5 \\
%(\bar D_5^*,b) & \bar P_1 & \bar P_2 & \bar P_3 & \bar P_4 & \bar P_5 
%\end{array}$$

The cascade product $B^*\circ C$ covers $A$.
Further decomposition of $C$ yields the desired covering of $A$ as stated in the above result.
\end{example}

\end{document}